\documentclass[11pt,a4paper]{article}
\usepackage[utf8]{inputenc}

\RequirePackage[l2tabu,orthodox]{nag}

\usepackage[ruled,vlined,linesnumbered]{algorithm2e}
\usepackage{amsmath}
\usepackage{amssymb}
\usepackage{amsthm}
\usepackage{array}
\usepackage{authblk}
\usepackage[english]{babel}
\usepackage{bigdelim}
\usepackage{booktabs}
\usepackage{cite}
\usepackage{float}
\usepackage{fullpage}
\usepackage{graphicx}
\usepackage[colorlinks,bookmarksopen,bookmarksnumbered,citecolor=red,urlcolor=red,unicode]{hyperref}
\usepackage{longtable}
\usepackage[bb=boondox]{mathalfa}
\usepackage{microtype}
\usepackage{multirow}
\usepackage[autolanguage]{numprint}
\usepackage{paralist}
\usepackage{rotating}
\usepackage{siunitx}
\usepackage{standalone}
\usepackage{subfigure}
\usepackage{tikz}
\usetikzlibrary{positioning,calc}
\usepackage{textcomp}
\usepackage{todonotes}
\usepackage{url}
\usepackage{wrapfig}
\usepackage{xspace}

\newtheorem{definition}{Definition}
\newtheorem{theorem}[definition]{Theorem}
\newtheorem{proposition}[definition]{Proposition}

\newtheorem{lemma}[definition]{Lemma}

\newcommand{\Dn}{\mathcal{D}^n}

\newcommand{\degree}[2]{\textnormal{deg}_{\textnormal{#1}}^{\textnormal{#2}}}
\newcommand{\len}{\textnormal{len}}
\newcommand{\absmass}{\textnormal{mass}^{\textnormal{abs}}}
\newcommand{\mass}{\textnormal{mass}}
\newcommand{\width}{\textnormal{width}}
\newcommand{\shape}[1]{\textnormal{sh}^{\textnormal{#1}}}
\newcommand{\erdos}{Erdős-Rényi\xspace}
\newcommand{\sameprob}{\texttt{sameprob}\xspace}
\newcommand{\samepred}{\texttt{samepred}\xspace}
\newcommand{\layrprob}{\texttt{layrprob}\xspace}
\newcommand{\layrpred}{\texttt{layrpred}\xspace}
\newcommand{\problem}{$P|p_j=1,prec|C_{\max}$\xspace}
\newcommand{\Expect}{\mathbb{E}}
\newcommand{\Prob}{\mathbb{P}}

\newcommand{\mynode}[3]{\node[circle,draw,thick] (#1) at (#2,#3){};}

\begin{document}

\title{A Comparison of Random Task Graph Generation Methods for
  Scheduling Problems}


\author[1]{Louis-Claude Canon}
\author[1]{Mohamad El Sayah}
\author[1]{Pierre-Cyrille Héam}
\affil[1]{FEMTO-ST Institute, CNRS, Univ. Bourgogne Franche-Comté, France}

\maketitle

\begin{abstract}
  How to generate instances with relevant properties and without bias
  remains an open problem of critical importance for a fair comparison
  of heuristics.
  In the context of scheduling with precedence constraints, the
  instance consists of a task graph that determines a partial order on
  task executions.
  To avoid selecting instances among a set populated mainly with
  trivial ones, we rely on properties that quantify the
  characteristics specific to difficult instances.
  Among numerous identified such properties, the \emph{mass} measures
  how much a task graph can be decomposed into smaller ones.
  This property, together with an in-depth analysis of existing random
  task graph generation methods, establishes the sub-exponential
  generic time complexity of the studied problem.
  Empirical observations on the impact of existing generation methods
  on scheduling heuristics concludes our study.
\end{abstract}

\section{Introduction}

How to correctly evaluate the performance of computing systems has
been a central question since several decades\cite{jain1990art}.
Among the arsenal of available evaluation methods, relying on random
instances allows comparing strategies in a large variety of
situations.
However, random generation methods are prone to bias, which prevents a
fair empirical assessment.
It is thus crucial to provide guarantees on the random distribution of
generated instances by ensuring, for instance, a uniform selection of
any instance among all possible ones.
Yet, for some problems, such uniformly generation instances are easy
to solve and thus uninteresting.
For instance, in uniformly distributed random graphs, the probability
that the diameter is 2 tends exponentially to 1 as the size of the
graph tends to infinity\cite{DBLP:journals/jsyml/Fagin76}.
Studying the problem characteristics to constrain the uniform
generation on a category of instances is thus critical.

In the context of parallel systems, instances for numerous
multiprocessor scheduling problems contain the description of an
application to be executed on a platform\cite{leung2004handbook}.
This study focuses on scheduling problems requiring a Directed Acyclic
Graph (DAG) as part of the input.
Such a DAG represents a set of tasks to be executed in a specific
order given by precedence constraints: the execution of any task
cannot start before all its predecessors have completed their
executions.
Scheduling a DAG on a platform composed of multiple processors
consists in assigning each task to a processor and in determining a
start time for each task.
While this work studies the DAG structure for several scheduling
problems, it illustrates and analyzes existing generation methods in
light of a specific problem with unitary costs and no communication.
This simple yet difficult problem emphasizes the effect of the DAG
structure on the performance of scheduling heuristics.

Some pathological instances are straightforward to solve.
For instance, if the width (i.e.\ maximum number of tasks that may be
run in parallel) is lower than the number of processors, then the
problem can be solved in polynomial time.
To avoid such instances, multiple DAG properties are proposed and
analyzed.
In particular, the mass measures the degree to which an instance can
be decomposed into smaller independent sub-instances.
In the absence of communication, this property has an impact on
scheduling algorithms.
The purpose of this work is to identify such properties to determine
how the uniform generation of DAGs should be constrained and how
existing generation methods perform relatively to these properties.
As a major contribution of this work, we determine the generic time
complexity to be sub-exponential for uniform instances for a large
class of scheduling problems (i.e.\ those that can be decomposed into
smaller problems).

After exposing related works in Section~\ref{section:related},
Section~\ref{section:background} lists DAG properties and covers
scheduling and random generation concepts.
Section~\ref{sec:analys-spec-dag} motivates the focus on a selection
of properties by analyzing all the proposed DAG properties on a set of
special DAGs.
Section~\ref{sec:existinggeneration} provides an in-depth analysis of
existing random DAG generation methods supported by consistent
empirical observations.
Finally, Section~\ref{sec:eval-sched-algor} studies the impact of
these methods and the DAG properties on scheduling heuristics.
The algorithms are implemented in R and Python and the related code,
data and analysis are available in\cite{figshare}.

\section{Related Work}\label{section:related}

\subsection{Analysis of Generation Methods}

Our approach is similar to the one followed in\cite{cordeiro2010a}
and\cite{martinez2018b}, which consists in studying the properties of
randomly generated DAGs before comparing the performance of scheduling
heuristics.
In\cite{cordeiro2010a}, three properties are measured and analyzed for
each studied generation method: the length of the longest path, the
distribution of the output degrees and the number of edges.
We describe 15 such properties in Table~\ref{tab:properties}.
They consider five random generation methods (described in this
section and Section~\ref{sec:existinggeneration}): two variants of
the \erdos algorithm, one layer-by-layer variant, the random orders
method and the Fan-in/Fan-out method.
Finally, for each generation method, the paper compares the
performance of four scheduling heuristics.
The results are consistent with the observations done in
Section~\ref{sec:existinggeneration}
(Figures~\ref{fig:erdos_proba},~\ref{fig:poset_perm}
and~\ref{fig:layer_layer}) for the length and the number of edges.
A similar approach is undertaken in\cite{martinez2018b}.
First, three characteristics are considered: the number of vertices in
the critical path, the width (or maximum parallelism) and the density
of the DAG in terms of edges.
These characteristics are studied on DAGs generated by two main
approaches (the \erdos algorithm and a MCMC approach) with sizes
between 5 and 30 vertices.
Finally, although no DAG property is studied, scheduling heuristics
are compared using a variety of random and non-random DAGs
in\cite{kwok1999a}.

We describe below generation tools, data sets and random generation
methods.

\subsection{Generation Tools}

Many tools have been proposed in the literature to generate DAGs in
the context of scheduling in parallel systems.
TGFF (Task Graphs For
Free)\footnote{\url{http://ziyang.eecs.umich.edu/projects/tgff/index.html}}
is the first tool proposed for this purpose\cite{dick1998a}.
This tool relies on a number of parameters related to the task graph
structure: maximum input and output degrees of vertices, average for
the minimum number of vertices, etc.
The task graph is constructed by creating a single-vertex graph and
then incrementally augmenting it.
This approach randomly alternates between two phases until the number
of vertices in the graph is greater than or equal to the minimum
number of vertices: the expansion of the graph and its contraction.
The main goal of TGFF is to gain more control over the input and
output degrees of the tasks.

DAGGEN\footnote{\url{https://github.com/frs69wq/daggen}} was later
proposed to compare heuristics for a specifc problem\cite{dutot2009a}.
This tool relies on a layer-by-layer approach with five parameters:
the number of vertices, a width and regularity parameters for the
layer sizes, and a density and jump parameters for the connectivity of
the DAG\@.
The number of elements per each layer is uniformly drawn in an
interval centered around an average value determined by the width
parameter and with a range determined by the regularity parameter.
Lastly, edges are added between layers separated by a maximum number
of layers determined by the jump parameter (edges only connect
consecutive layers when this parameter is one).
For each vertex, a uniform number of predecessors is added between one
and a maximum value determined by the density parameter.

GGen\footnote{\url{https://github.com/perarnau/ggen}} has been
proposed to unify the generation of DAGs by integrating existing
methods\cite{cordeiro2010a}.
The tool implements two variants of the \erdos algorithm, one
layer-by-layer variant, the random orders method and the
Fan-in/Fan-out method.
It also generates DAGs derived from classical parallel algorithms such
as the recursive Fibonacci function, the Strassen multiplication
algorithm, the Cholesky factorization, etc.

The Pegasus workflow
generator\footnote{\url{https://confluence.pegasus.isi.edu/display/pegasus/WorkflowGenerator}}
can be used to generate DAGs from several scientific
applications\cite{juve2013a} such as Montage, CyberShake, Broadband,
etc.
XL-STaGe\footnote{\url{https://github.com/nizarsd/xl-stage}} produces
layer-by-layer DAGs using a truncated normal distribution to
distribute the vertices to the layers\cite{campos2016a}.
This tool inserts edges with a probability that decreases as the
number of layers between two vertices increases.
A tool named
RandomWorkflowGenerator\footnote{\url{https://github.com/anubhavcho/RandomWorkflowGenerator}}
implements a layer-by-layer variant\cite{gupta2017a}.
Other tools have also been proposed but are no longer available as of
this writing: DAGEN\cite{amalarethinam2011a},
RTRG\footnote{\url{http://users.ecs.soton.ac.uk/ras1n09/rtrg/index.html}
  (unavailable as of this writing)}\cite{shafik2012a},
MRTG\cite{ashish2016a}.

Finally, other fields such as electronic circuit design or dataflow
also use DAGs.
In this last field, however, requirements differ: the acyclicity is no
longer relevant, while ensuring a strong connectedness is important.
Two noteworthy generators have been proposed SDF$^3$ inspired from
TGFF\footnote{\url{http://www.es.ele.tue.nl/sdf3/}}\cite{stuijk2006a}
and
Turbine\footnote{\url{https://github.com/bbodin/turbine}}\cite{bodin2014a}.

\subsection{Instance Sets}

The STG (Standard Task Graph)
set\footnote{\url{http://www.kasahara.elec.waseda.ac.jp/schedule/}}
has been specifically proposed for parallel systems\cite{tobita2002a}
and is frequently used to compare scheduling
heuristics\cite{aggarwal2005genetic,davidovic2012bee}.
The DAG structures of STG relies on four different methods.
Two methods, \sameprob and \samepred, rely on the \erdos algorithm,
while the other two, \layrprob and \layrpred, constitute
layer-by-layer variants.
A connection probability is given to \sameprob and \layrprob, while an
average number of predecessors is given to \samepred and \layrpred.
With these last two methods, the parameter is apparently converted to
a connection probability inferred from the size of the DAG\@.
Any layer-by-layer variant proceeds by first distributing vertices
into layers such that the average layer size is 10.
Then, edges between any pair of vertices from distinct layers are
added from top to bottom according to the connectivity parameter.
The size of the DAGs varies from 50 to \numprint{5000}.
For each size, the data set contains 15 instances for each combination
of a method among the four ones and a value for the connectivity
parameter among three possible ones (leading to 180 instances).
Both layer-by-layer variants do not guarantee that the layer of any
vertex equals its depth.
As a consequence, the length is not necessarily $\frac{n}{10}+2$ (2
dummy vertices are always added) where $n$ is the number of
vertices\footnote{This is the case for the instance
  \texttt{rand0038.stg} for size 50.}
and this problem becomes more apparent with large DAGs generated by
\layrpred because there are not enough inserted edges to ensure the
layered structure.
The STG set also contains costs and real DAGs such as robot control,
sparse matrix solver and SPEC fpppp program.

PSPLIB\footnote{\url{http://www.om-db.wi.tum.de/psplib/}} contains
difficult instances for RCPSP (Resource-Constrained Project Scheduling
Problems)\cite{kolisch1995a}, a scheduling problem in the field of
project management.
Finally, in the graph drawing context, a set of 112 real-life graphs
were proposed\footnote{\url{ftp://infokit.dis.uniromal.it/public/}
  (unavailable as of this writing)}\cite{battista1997a} but are no
longer available.

In addition to those implemented in GGen and the ones in STG, other
DAGs from real-cases can be used such as the LU
decomposition\cite{lord1983a}, the parallel Gaussian elimination
algorithm\cite{cosnard1988a}, the parallel Laplace equation
algorithm\cite{wu1990a}, the mean value analysis
(MVA)\cite{almeida1992a}, which has a diamond-like structure, the FFT
algorithm\cite{cormen2009a}, which has a butterfly structure, the
QR factorization, etc.

\subsection{Layer-by-Layer Methods}
\label{sec:layer-layer-method}

The layer-by-layer method was first proposed by\cite{adam1974a} but
popularized later by the introduction of the STG data
set\cite{tobita2002a}.
This method produces DAGs in which vertices are distributed in layers
and vertices belonging to the same layer are independent.
The method consists in three steps: determining the number of layers;
distributing the vertices to the layers; connecting the vertices from
different layers.
In most proposed methods, there is at least one parameter for each
step.
For instance, the shape parameter controls the number of layers and is
related to the ratio of $\sqrt{n}$ to the number of
layers\cite{topcuoglu2002a,ilavarasan2007a,gupta2017a}.

The number of layers can be drawn from a parameterized uniform
distribution\cite{adam1974a,topcuoglu2002a,ilavarasan2007a,saovapakhiran2011a},
given as a parameter\cite{cordeiro2010a,gupta2017a} or generated in a
non-parameteri\-zed way\cite{ahmad1998a,tobita2002a,campos2016a}.

Similarly, vertices can be distributed by generating a number of
vertices at each layer with a parameterized uniform
distribution\cite{adam1974a,topcuoglu2002a,ilavarasan2007a,dutot2009a,saovapakhiran2011a},
by selecting a layer for each vertex with a parameterized normal
distribution\cite{campos2016a}, by using a balls into bins
approach\cite{cordeiro2010a,gupta2017a} or in a non-parameterized
way\cite{ahmad1998a}.
Note that generating a uniform number of vertices per layer may lead
to a different number of vertices $n$ than expected.
Also, using a balls into bins strategy may lead to empty layers.

Finally, the connection between vertices can depend on a connection
probability\cite{tobita2002a,dutot2009a,cordeiro2010a,saovapakhiran2011a,campos2016a}
or an average number of predecessors or successors for each
vertex\cite{adam1974a,tobita2002a,topcuoglu2002a}.
Although vertices in the same layer may have different depth (e.g.\
this occurs in the STG data set), adding specific edges prevents this
situation\cite{dutot2009a,gupta2017a}.
The layer-by-layer approach can also lead to DAGs with multiple
connected components except for\cite{adam1974a}.
Finally, some methods allow edges between non-consecutive
layers\cite{adam1974a,tobita2002a,cordeiro2010a}, while others limit
them\cite{dutot2009a,campos2016a,gupta2017a}.

\subsection{Uniform Random Generation}

Many works address the problem of randomly generating DAGs
with a known distribution.
Uniform random generation of DAGs can be done using counting
approaches\cite{roblab} based on generating functions.
Many exisiting methods have been developped in the literature and the
most important ones are described in
Section~\ref{sec:existinggeneration}.

While previous uniform approaches consider only the size of the DAG
$n$ as a parameter, other studies have proposed to generate directed
graphs from a prescribed degree
sequence\cite{milo2003a,karrer2009a,ajwani2013a}.
A uniform method is proposed in\cite{milo2003a} but may produce cyclic
graphs.
In contrast, the method proposed in\cite{karrer2009a} forbids
cyclicity but has no uniformity guarantee.
Last, in the context of sensor streams, several methods has been
proposed\cite{ajwani2013a} to generate DAGs with a prescribed degree
distribution.

Finally, a multitude of related approaches has been proposed but are
discarded in this study because of their specificity.
For instance, specific structures may be used to assess the
performance of scheduling methods\cite{li2013a,canon2018a} or special
DAGs with known optimal solutions relatively to a given platform may
also be built\cite{kwok1999a,oppermann2018a}.

\section{Background}
\label{section:background}

\subsection{Directed Acyclic Graphs}
\label{sec:direct-acycl-graphs}
All graphs considered throughout this paper are finite.  A
\emph{directed graph} is a pair $(V,E)$ where $V$ is a finite set of
\emph{vertices} and $E\subseteq V\times V$ is the set of \emph{edges}. A
\emph{path} is a finite sequence of consecutive edges, that is a
sequence of the form $(v_1,v_2),(v_2,v_3),\ldots,(v_{k-1},v_k)$; $k$
is the \emph{length} of the path, i.e.\ the number of vertices on this
path.

The \emph{output degree} of a vertex $v$ is the cardinal of the set
$\{(v,w)\mid w\in V,\ (v,w)\in E\}$. Similarly the \emph{input
  degree} of a vertex $v$ is the cardinal of the set $\{(w,v)\mid
w\in V,\ (w,v)\in E\}$.  The \emph{output} (resp.\ \emph{input}) \emph{degree} of a
directed graph is the maximum value of the output (resp.\ input) degrees
of its vertices.
The degree of a vertex is the sum of its input and output degrees.

A directed graph is \emph{acyclic} (DAG for short) if there is no path
of strictly positive length $k$ such that $v_1=v_k$ (with the above
notation).
Let $\Dn$ be the set of all DAGs whose set of vertices is
$\{1,2,\ldots,n\}$.
In a DAG, if
$(v,w)$ is an edge, $v$ is a \emph{predecessor} of $w$ and $w$ a
\emph{successor} of $v$.

In a DAG $D$ with $n$ vertices, all paths have a length less than or equal
to $n$. The \emph{length} of a DAG is defined as the maximum length of
a path in this DAG\@. The \emph{depth} of a vertex $v$ in a DAG is
inductively defined by: if $v$ has no predecessor, then its depth is
$1$; otherwise, the depth of $v$ is one plus the maximum depth of its
predecessors. The \emph{shape decomposition} of a DAG is the tuple
$(X_1,X_2,\ldots,X_k)$ where $X_i$ is the set of vertices of depth
$i$. Note that $k$ is the length of the DAG\@.
The \emph{shape} of the DAG is the tuple $(|X_1|,\ldots,|X_k|)$.
The maximum (resp.\ minimum) value of the $|X_i|$ is called the
\emph{maximum shape} (resp.\ \emph{minimum shape}) of the DAG\@.
Computing the shape decomposition and the shape of a DAG is easy.
If $|X_i|=1$, the unique vertex of $X_i$ is called a \emph{bottleneck
  vertex}.
A \emph{block} is a subset of vertices of the form
$\cup_{i<j<i+\ell} X_j$ with $\ell>1$ where $X_i$ is either a
singleton or $i=0$, $X_{i+\ell}$ is either a singleton or
$i+\ell=k+1$, and for each $i< j<i+\ell$, $|X_j|\neq 1$.
We denote by $\absmass(B)$ the cardinal of $B=\cup_{i<j<i+\ell} X_j$
and by $\absmass(D)=\max\{\absmass(B)\mid B \text{ is a block}\}$ the
\emph{absolute mass} of $D$.
The \emph{relative mass}, or simply the \emph{mass}, is given by
$\mass(D)=\frac{\absmass(D)}{n}$.

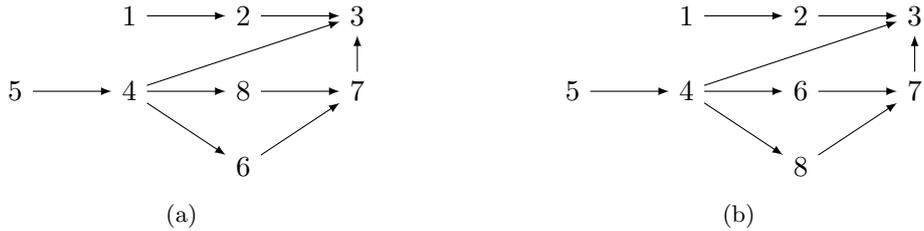
\begin{figure}
  \centering
    \subfigure[\label{fig:DAG:examplea}]{
      \begin{tikzpicture}
        \node(1) at (0,0) {1};
        \node(2) at (1.5,0) {2};
        \node(3) at (3,0) {3};
        \node(5) at (-1.5,-1) {5};
        \node(4) at (0,-1) {4};
        \node(8) at (1.5,-1) {8};
        \node(6) at (1.5,-2) {6};
        \node(7) at (3,-1) {7};
        \path[draw,->,>=latex] (1) -- (2);
        \path[draw,->,>=latex] (2) -- (3);
        \path[draw,->,>=latex] (5) -- (4);
        \path[draw,->,>=latex] (4) -- (8);
        \path[draw,->,>=latex] (4) -- (6);
        \path[draw,->,>=latex] (8) -- (7);
        \path[draw,->,>=latex] (6) -- (7);
        \path[->,draw,>=latex] (4) to (3);
        \path[->,draw,>=latex] (7) to (3);
      \end{tikzpicture}
    }\hspace{2cm}
    \subfigure[]{
      \begin{tikzpicture}
        \node(1) at (0,0) {1};
        \node(2) at (1.5,0) {2};
        \node(3) at (3,0) {3};
        \node(5) at (-1.5,-1) {5};
        \node(4) at (0,-1) {4};
        \node(8) at (1.5,-1) {6};
        \node(6) at (1.5,-2) {8};
        \node(7) at (3,-1) {7};
        \path[draw,->,>=latex] (1) -- (2);
        \path[draw,->,>=latex] (2) -- (3);
        \path[draw,->,>=latex] (5) -- (4);
        \path[draw,->,>=latex] (4) -- (8);
        \path[draw,->,>=latex] (4) -- (6);
        \path[draw,->,>=latex] (8) -- (7);
        \path[draw,->,>=latex] (6) -- (7);
        \path[->,draw,>=latex] (4) to (3);
        \path[->,draw,>=latex] (7) to (3);
      \end{tikzpicture}
    }
  \caption{Examples of DAGs.}\label{fig:DAG:example}
\end{figure}

For example, the DAG on Fig.~\ref{fig:DAG:examplea} has for shape
decomposition the tuple
$(\{1,5\},\{2,4\},\allowbreak\{6,8\},\{7\},\{3\})$ and for shape
the tuple $(2,2,2,1,1)$. A longest path is
$(5,4),(4,6),(6,7),(7,3)$. It has two bottleneck vertices $7$ and
$3$. Its absolute mass is $2+2+2+1=7$.

In a DAG, two distinct vertices $v$ and $w$ are \emph{incomparable} if
there is neither a path from $v$ to $w$, nor from $w$ to $v$. The
\emph{width} of a graph is the maximum size of the subset of vertices
whose elements are pair-wise incomparable. Since vertices of same
depth are incomparable, the maximum shape of a DAG is less than or equal to
its width.
The width is also the size of the largest antichain, which can be
computed in polynomial time using Dilworth's theorem and a technique
developed by Ford and Fulkerson\cite{ford2015flows}.
The methodology is conjectured to have a time complexity of
$O(n^{5/2})$\cite{plotnikov2007experimental}.
In some cases (for instance the \emph{comb} DAG, see
Section~\ref{sec:analys-spec-dag}),
the width can be much larger than the maximum shape.
Table~\ref{tab:comparison} compares the width and the maximum shape on
the DAGs obtained with two random generators explored in
this paper.

\begin{table}
  \centering
  \begin{tabular}{ccc}
    \toprule
    $n$ & \erdos & Uniform\\
    \midrule
    10 & 2.95 -- 0.34 -- 2 & 2.35 -- 0.09 -- 1\\
    20 & 3.52 -- 0.45 -- 2 & 2.77 -- 0.14 -- 1\\
    30 & 3.62 -- 0.46 -- 2 & 3.13 -- 0.23 -- 1\\
    \bottomrule
  \end{tabular}
  \caption{\label{tab:comparison}Comparison of width and maximum shape
    of randomly generated DAGs with different methods: ``\erdos'' for the
    so-called algorithm with parameter $p=0.5$ (see
    Section~\ref{sec:rand-gener-triang}) and ``Uniform'' for the
    recursive random generator (see Section~\ref{sec:recursive}).
    Reported numbers $x - y - z$ correspond respectively to the
    average width, the average difference between width and shape
    width, and the maximum difference pointed out.
    Each experiment is performed by sampling 100 DAGs.
  }
\end{table}

Two DAGs $(V_1,E_1)$ and $(V_2,E_2)$ are isomorphic, denoted
$(V_1,E_1)\sim(V_2,E_2)$, if there exists a bijective map $\varphi$
from $V_1$ to $V_2$ such that $(x,y)\in E_1$ iff
$(\varphi(x),\varphi(y))\in E_2$. The relation $\sim$ is an
equivalence relation.
Intuitively, two DAGs are isomorphic if they are equal up to
vertices names.  For example, the DAGs on Fig.~\ref{fig:DAG:example} are
isomorphic.

The \emph{transitive reduction} of a DAG $D$\cite{aho1972transitive}
is the DAG $D^T$ for which: $D^T$ has a directed path between $u$ and
$v$ iff $D$ has a directed path between $u$ and $v$; there is no graph
with fewer edges than $D^T$ that satisfies the previous property.
Intuitively, this operation consists in removing redundant edges.
The \emph{reversal} of a DAG $D$ is the DAG $D^R$ for which there is
an edge between $u$ and $v$ iff there is an edge between $v$ and $u$
in $D$.
Intuitively, this operation consists in reversing the DAG\@.

Finally, Table~\ref{tab:properties} presents some of the DAG
properties that may impact the performance of scheduling algorithms.
We discard the minimum input and output degrees because they are
always $\degree{in}{min}=\degree{out}{min}=0$.
We also discard the mean input and output degrees because they are
always equal to half the mean degree
($\degree{in}{mean}=\degree{out}{mean}=\frac{\degree{}{mean}}{2}$).
For all nine edge-related properties ($m$ and the degree-based
properties) applied to a DAG $D$, we can also compute them on the
transitive reduction $D^T$.
The vertex-related properties ($n$, the width and the shape-based
ones) remain the same on the transitive reduction.
For all seven shape-based properties on a DAG $D$, we can also compute
them on the reversal $D^R$.
The edge-related properties remain the same through the reversal with
the inversion of $\degree{in}{max}$ and $\degree{in}{sd}$
with $\degree{out}{max}$ and $\degree{out}{sd}$,
respectively.
Finally, some of these properties are related:
$n\times\degree{}{mean}=\frac{m}{2}$ and $\len\times\shape{mean}=n$.

\begin{table}[ht]
  \centering
  \begin{tabular}{rm{0.6\columnwidth}}
    \toprule
    Symbol & Definition\\
    \midrule
    $n$ & number of vertices\\
    $m$ & number of edges\\
    $\degree{}{max} (\degree{in}{max},\degree{out}{max}$) & maximum (input, output) degree\\
    $\degree{}{min}$ & minimum degree\\
    $\degree{}{mean}$ & mean (input, output) degree\\
    $\degree{}{sd} (\degree{in}{sd},\degree{out}{sd})$ & standard deviation of the (input, output) degrees\\
    $\len$ or $k$ & length (also called height, number of levels, longest
                    path or critical path length)\\
    $\width$ & width\\
    $\shape{max}$ & maximum shape\\
    $\shape{min}$ & minimum shape\\
    $\shape{mean}$ & mean shape (parallelism in\cite{tobita2002a})\\
    $\shape{sd}$ & standard deviation of the shape\\
    $\shape{1}$ & number of source vertices (vertices with null input degree)\\
    $\shape{k}$ & last element of the shape\\
    $\mass$ & (relative) mass\\
    \midrule
    $p$ & connectivity probability\\
    $K$ & number of permutations (for the random orders method in Section~\ref{sec:random-orders})\\
    $P$ & set of processors\\
    \bottomrule
  \end{tabular}
  \caption{\label{tab:properties}List of DAG properties and other
    notations.
    When necessary, we specify on which DAG a property is measured
    (e.g.\ $m(D^T)$ for the number of edges in the transitive
    reduction of $D$).
  }
\end{table}

\subsection{Scheduling}
\label{sec:scheduling}

We consider a classic problem in parallel systems noted \problem in
Graham's notation\cite{graham79a}.
The objective consists in scheduling a set of tasks on homogeneous
processors such as to minimize the overall completion time.
The dependencies between tasks are represented by a precedence DAG
$(V,E)$ where $|V|=n$ is the number of tasks and $|E|=m$ the number of
edges.
Before starting the execution of a task, all its predecessors must
complete their executions.
The execution cost $p_j$ of task $j$ on any processor is unitary and
there is no cost on the edges (i.e.\ no communication).
A schedule defines on which processor and at which date each task
starts its execution such that no processor executes more than one
task at any time and all precedence constraints are met.
The problem consists in finding the schedule with the minimum
makespan, i.e.\ overall completion time before the first task starting
its execution and the last one completing its execution.

A possible schedule for the DAG of Figure~\ref{fig:DAG:examplea} on
two processors $P_1$ and $P_2$, assuming costs are unitary, consists
in starting executing tasks~1 and~2 on processor $P_1$ as soon as
possible (i.e.\ at times~0 and~1), while processor $P_2$ processes
tasks~5, 4, 8, 7 and 3 similarly.
The execution of task~6 follows the termination of task~2 on
processor $P_1$ to satisfy the precedence constraint of task~7.
The makespan of this schedule is 5.

This problem is strongly
NP-hard\cite{Ullman:75:NP-complete-scheduling}, while it is polynomial
when there are no precedence constraints ($P|p_j=1|C_{\max}$), which
means the difficulty comes from the dependencies.
Many polynomial heuristics have been proposed for this problem (see
Section~\ref{sec:eval-sched-algor}).
With specific instances, such heuristics may be optimal.
This is the case when the width does not exceed the number of
processors, which leads to a potentially large length.
Any task can thus start its execution as soon as it becomes available.
The problem is also polynomial when edges only belong to the critical
path (i.e.\ $m=\len-1$ and the width equal $n-\len+1$, which is large
when the length is small).
In this case, any heuristic prioritizing critical tasks and scheduling
all other tasks as soon as possible will be optimal.
This paper explores how DAG properties are impacted by the generation
method with the objective to control them to avoid easy instances.

Although this paper studies random DAGs with heuristics for the
specific problem \problem, generated DAGs can be used for any
scheduling problem with precedence constraints.
While avoiding specific instances depending on their width and length
is relevant for many scheduling problems, it is not necessary the case
for all of them.
For instance, with non-unitary processing costs, instances with large
width and small length are difficult because the problem is strongly
NP-Hard even in the absence of precedence constraints
($P||C_{\max}$)\cite{GareyJohnson:78:Strong-NP-completeness}.

\subsection{Mass and Scheduling}
\label{sec:mass-scheduling}

The proposed mass measure has a direct implication in this scheduling
context.
Consider a DAG $D=(V,E)$ whose minimum shape is $1$; there exists a
bottleneck vertex $v$ such that the shape of the DAG is of the form
$(X_1,\ldots,X_{\ell},\{v\},X_{\ell+1},\ldots,X_k)$. The scheduling
problem for $D$ can be decomposed into two subproblems, one for
the sub-DAG of $D$ whose set of vertices is $\{v\}\cup\bigcup_{i\leq
\ell}X_i$ and one for the sub-DAG of $D$ whose set of vertices is
$\{v\}\cup\bigcup_{i>
  \ell} X_i$. Using recursively this decomposition, the initial problem
can be decomposed into $n_c+1$ independent scheduling problems, where
$n_c$ is the number of  bottleneck
vertices.

Applying a brute force algorithm for the scheduling problems
computes the optimal results in a time $T\leq n_c T_m$, where $T_m$ is
the maximum time required to solve the problem on a DAG with $\absmass(D)$
vertices. Since exponential brute force exact approaches exist, it
follows that if $\absmass(D)=O(\log^k n)$ for a constant $k$, then an
optimal solution of the
scheduling problem can be computed in sub-exponential time. Consequently,
scheduling heuristics are irrelevant for task graph with logarithmic
absolute mass. 
Similarly, the same arguments work to claim that interesting instances
for the scheduling problem must have quite a large absolute mass (not
in $o(n)$).
It is therefore preferable to have instances with no or few bottleneck
vertices, that is a unitary mass.

The relevance of the mass property is limited to a specific class of
scheduling problems that contains all problems for which the instance
can be cut into independent instances.
While the mass is still relevant with non-unitary processing costs, it
is no longer the case when there are communication costs.

\subsection{Uniformity of the Random Generation}
\label{sec:unif-rand-gener}

\begin{table}[ht]
\centering
\begin{tabular}{cccc}
  \toprule
  Isom.\ classes & Matrices & ER & Labeling\\

  \midrule
  \begin{tikzpicture}
    \mynode{A}{0}{0};
    \mynode{B}{1}{0};
    \mynode{C}{2}{0};
  \end{tikzpicture}
  &
  $\left(\begin{matrix}0 & 0 \\ & 0\end{matrix}\right)$
    & $\frac{1}{8}$ & 1\\

  \begin{tikzpicture}
    \mynode{A}{0}{0};
    \mynode{B}{1}{0};
    \mynode{C}{2}{0};
    \path[->,thick,draw] (A) -- (B);
  \end{tikzpicture}
  &
  $\left(\begin{matrix}1 & 0 \\ & 0\end{matrix}\right)$
    $\left(\begin{matrix}0 & 1 \\ & 0\end{matrix}\right)$
      $\left(\begin{matrix}0 & 0 \\ & 1\end{matrix}\right)$
    & $\frac{3}{8}$ & 6\\

  \begin{tikzpicture}
    \mynode{A}{0}{0};
    \mynode{B}{1}{0};
    \mynode{C}{2}{0};
    \path[->,thick,draw] (B) -- (C);
    \path[->,thick,draw] (A) to [bend left] (C);
  \end{tikzpicture}
  &
  $\left(\begin{matrix}0 & 1 \\ & 1\end{matrix}\right)$
    & $\frac{1}{8}$ & 3\\

  \begin{tikzpicture}
    \mynode{A}{0}{0};
    \mynode{B}{1}{0};
    \mynode{C}{2}{0};
    \path[->,thick,draw] (A) -- (B);
    \path[->,thick,draw] (A) to [bend left] (C);
  \end{tikzpicture}
  &
  $\left(\begin{matrix}1 & 1 \\ & 0\end{matrix}\right)$
    & $\frac{1}{8}$ & 3\\

  \begin{tikzpicture}
    \mynode{A}{0}{0};
    \mynode{B}{1}{0};
    \mynode{C}{2}{0};
    \path[->,thick,draw] (A) -- (B);
    \path[->,thick,draw] (B) -- (C);
  \end{tikzpicture}
  &
  $\left(\begin{matrix}1 & 0 \\ & 1\end{matrix}\right)$
    & $\frac{1}{8}$ & 6\\

  \begin{tikzpicture}
    \mynode{A}{0}{0};
    \mynode{B}{1}{0};
    \mynode{C}{2}{0};
    \path[->,thick,draw] (B) -- (C);
    \path[->,thick,draw] (A) -- (B);
    \path[->,thick,draw] (A) to [bend left] (C);
  \end{tikzpicture}
  &
  $\left(\begin{matrix}1 & 1 \\ & 1\end{matrix}\right)$
    & $\frac{1}{8}$ & 6\\
  \bottomrule
\end{tabular}
\caption{\label{tab:dag3}DAGs with 3 vertices: there is one row for
  each isomorphism class.
  For each class, we report: all corresponding (upper triangular)
  adjacency matrices; the probability of generating such a DAG with
  the \erdos algorithm ($p=0.5$); and, the number of DAGs in each
  isomorphism class (i.e.\ the number of labelings).
}
\end{table}

This work focuses on the importance generating DAGs uniformly.
We discuss the notion of uniformity through the example with 3
vertices given in Table~\ref{tab:dag3}.
In this instance, there are six isomorphism classes (i.e.\ six
different unlabeled DAGs) for a total of 25 different (labeled) DAGs.
A generator is thus uniform up to isomorphism if it generates each
isomorphism class (or unlabelled DAGs) with a probability $\frac1{6}$
or uniform on all (labelled) DAGs if it generates each DAG with a
probability $\frac1{25}$.
We also say that we generate non-isomorphic DAGs in the former case.
Finally, when considering only transitive reductions, we discard the
complete DAG\@.
The probability to generate each of the remaining isomorphism classes
(resp.\ labeled DAGs) with a uniform generator becomes $\frac1{5}$
(resp.\ $\frac1{19}$).
This leads to four different uniformity definitions.

\section{Analysis of special DAGs}
\label{sec:analys-spec-dag}

\begin{table}
  \centering
  \begin{tabular}{m{4cm}m{6cm}>{\centering\arraybackslash}m{4cm}}
    \toprule
    Name & description & representation\\
    \midrule
    Empty ($D_{\textnormal{empty}}$) & no edge &
    \begin{tikzpicture}
      \mynode{A}{0}{0};
      \mynode{B}{1}{0};
      \mynode{C}{2}{0};
      \mynode{D}{3}{0};
    \end{tikzpicture}\\

    Complete ($D_{\textnormal{complete}}$) & maximum number of edges &
    \begin{tikzpicture}
      \mynode{A}{0}{0};
      \mynode{B}{1}{0};
      \mynode{C}{2}{0};
      \mynode{D}{3}{0};
      \path[->,thick,draw] (A) -- (B);
      \path[->,thick,draw] (B) -- (C);
      \path[->,thick,draw] (C) -- (D);
      \path[->,thick,draw] (B) to[bend left] (D);
      \path[->,thick,draw] (A) to[bend right] (C);
      \path[->,thick,draw] (A) to[bend right] (D);
    \end{tikzpicture}\\

    Chain ($D_{\textnormal{chain}}$) & transitive reduction of the
    complete DAG &
    \begin{tikzpicture}
      \mynode{A}{0}{0};
      \mynode{B}{1}{0};
      \mynode{C}{2}{0};
      \mynode{D}{3}{0};
      \path[->,thick,draw] (A) -- (B);
      \path[->,thick,draw] (B) -- (C);
      \path[->,thick,draw] (C) -- (D);
    \end{tikzpicture}\\

    Complete binary tree ($D_{\textnormal{out-tree}}$) & each
    non-leaf/non-root vertex has a unique predecessor and two
    successors &
    \begin{tikzpicture}
      \mynode{A}{0}{0};
      \mynode{B}{1}{0.5};
      \mynode{C}{1}{-0.5};
      \mynode{D}{2}{0.25};
      \mynode{E}{2}{0.75};
      \mynode{F}{2}{-0.25};
      \mynode{G}{2}{-.75};
      \path[->,thick,draw] (A) -- (B);
      \path[->,thick,draw] (A) -- (C);
      \path[->,thick,draw] (B) -- (D);
      \path[->,thick,draw] (B) -- (E);
      \path[->,thick,draw] (C) -- (F);
      \path[->,thick,draw] (C) -- (G);
    \end{tikzpicture}\\


    Comb ($D_{\textnormal{comb}}$) & a chain where each non-leaf
    vertex has an additional leaf successor &
    \begin{tikzpicture}
      \mynode{A}{0}{-1};
      \mynode{B}{1}{-1};
      \mynode{C}{2}{-1};
      \mynode{D}{3}{-1};
      \mynode{E}{1}{-0.5};
      \mynode{F}{2}{-0.5};
      \mynode{G}{3}{-0.5};
      \path[->,thick,draw] (A) -- (B);
      \path[->,thick,draw] (B) -- (C);
      \path[->,thick,draw] (C) -- (D);
      \path[->,thick,draw] (A) -- (E);
      \path[->,thick,draw] (B) -- (F);
      \path[->,thick,draw] (C) -- (G);
    \end{tikzpicture}\\

    Complete bipartite ($D_{\textnormal{bipartite}}$) & $\frac{n}{2}$ vertices
    connected to $\frac{n}{2}$ vertices &
    \begin{tikzpicture}
      \mynode{A}{0}{0};
      \mynode{B}{1}{0};
      \mynode{C}{2}{0};
      \mynode{D}{0}{-1};
      \mynode{E}{1}{-1};
      \mynode{F}{2}{-1};
      \path[->,thick,draw] (A) -- (D);
      \path[->,thick,draw] (A) -- (E);
      \path[->,thick,draw] (A) -- (F);
      \path[->,thick,draw] (B) -- (D);
      \path[->,thick,draw] (B) -- (E);
      \path[->,thick,draw] (B) -- (F);
      \path[->,thick,draw] (C) -- (D);
      \path[->,thick,draw] (C) -- (E);
      \path[->,thick,draw] (C) -- (F);
    \end{tikzpicture}\\

    Complete layer-by-layer square ($D_{\textnormal{square}}$) &
    similar to the complete bipartite with $\sqrt{n}$ layers of size
    $\sqrt{n}$ &
    \begin{tikzpicture}
      \mynode{A}{0}{0.6};
      \mynode{B}{0}{0};
      \mynode{C}{0}{-0.6};
      \mynode{D}{1}{0.6};
      \mynode{E}{1}{0};
      \mynode{F}{1}{-0.6};
      \mynode{G}{2}{0.6};
      \mynode{H}{2}{0};
      \mynode{I}{2}{-0.6};
      \path[->,thick,draw] (A) -- (D);
      \path[->,thick,draw] (A) -- (E);
      \path[->,thick,draw] (A) -- (F);
      \path[->,thick,draw] (B) -- (D);
      \path[->,thick,draw] (B) -- (E);
      \path[->,thick,draw] (B) -- (F);
      \path[->,thick,draw] (C) -- (D);
      \path[->,thick,draw] (C) -- (E);
      \path[->,thick,draw] (C) -- (F);
      \path[->,thick,draw] (D) -- (G);
      \path[->,thick,draw] (D) -- (H);
      \path[->,thick,draw] (D) -- (I);
      \path[->,thick,draw] (E) -- (G);
      \path[->,thick,draw] (E) -- (H);
      \path[->,thick,draw] (E) -- (I);
      \path[->,thick,draw] (F) -- (G);
      \path[->,thick,draw] (F) -- (H);
      \path[->,thick,draw] (F) -- (I);
    \end{tikzpicture}\\

    Complete layer-by-layer triangular ($D_{\textnormal{triangular}}$) &
    similar to the complete layer-by-layer square but the size of each
    new layer increases by~1 &
    \begin{tikzpicture}
      \mynode{A}{0}{0};
      \mynode{B}{1}{0};
      \mynode{C}{1}{0.6};
      \mynode{D}{2}{0.6};
      \mynode{E}{2}{1.2};
      \mynode{F}{2}{0};
      \path[->,thick,draw] (A) -- (B);
      \path[->,thick,draw] (A) -- (C);
      \path[->,thick,draw] (B) -- (F);
      \path[->,thick,draw] (B) -- (D);
      \path[->,thick,draw] (B) -- (E);
      \path[->,thick,draw] (C) -- (F);
      \path[->,thick,draw] (C) -- (D);
      \path[->,thick,draw] (C) -- (E);
    \end{tikzpicture}\\

    \bottomrule
  \end{tabular}
  \caption{\label{tab.DAG}Special DAGs.
    The number of vertices $n$ is assumed to be a power of two minus
    one for the tree, odd for the comb, even for the bipartite, a square
    for the square and a triangular number for the triangular (one of
    the form $1+2+3+\cdots+k$).}
\end{table}

To analyse the properties described in the previous section, we
introduce in Table~\ref{tab.DAG} a collection of special DAGs.
The first three DAGs ($D_{\textnormal{empty}}$,
$D_{\textnormal{complete}}$ and $D_{\textnormal{chain}}$) constitutes
extreme cases in terms of precedence.
The next two DAGs ($D_{\textnormal{out-tree}}$ and
$D_{\textnormal{comb}}$), to which we can add the reversal of the
complete binary tree
($D_{\textnormal{in-tree}}=D_{\textnormal{out-tree}}^R$), are examples
of binary tree DAGs.
The last three DAGs ($D_{\textnormal{bipartite}}$,
$D_{\textnormal{square}}$ and $D_{\textnormal{triangular}}$) are
denser with more edges and with a compromise between the length and
the width for the last two DAGs.

\begin{table}[ht]
  \centering
  \begin{tabular}{m{0.08\columnwidth}ccccccccc}
    \toprule
    DAG & $m$ & $\degree{}{max}$ & $\degree{in}{max}$ & $\degree{out}{max}$
    & $\degree{}{min}$ & $\degree{}{mean}$ & $\degree{}{sd}$ & $\degree{in}{sd}$
    & $\degree{out}{sd}$\\
    \midrule
    $D_{\textnormal{empty}}$ & 0 & 0 & 0 & 0 & 0 & 0 & 0 & 0 & 0\\
    $D_{\textnormal{complete}}$ & $\frac{n^2}{2}$
    & $n$ & $n$ & $n$ & $n$ & $n$ & 0 & $\frac{n}{\sqrt{12}}$ & $\frac{n}{\sqrt{12}}$\\
    $D_{\textnormal{chain}}$ & $n$ & $2$ & $1$ & $1$ & $1$ & $2$
    & $\sqrt{\frac{2}{n}}$ & $\frac{1}{\sqrt{n}}$ & $\frac{1}{\sqrt{n}}$\\
    $D_{\textnormal{out-tree}}$ $D_{\textnormal{comb}}$ & $n$ & $3$
    & $1$ & $2$ & $1$ & $2$ & $1$
    & $\frac{1}{\sqrt{n}}$ & $1$\\
    $D_{\textnormal{in-tree}}$ $D_{\textnormal{comb}}^R$ & $n$ & $3$ & $2$
    & $1$ & $1$ & $2$ & $1$
    & $1$ & $\frac{1}{\sqrt{n}}$\\
    $D_{\textnormal{bipartite}}$ & $\frac{n^2}{4}$ & $\frac{n}{2}$ & $\frac{n}{2}$
    & $\frac{n}{2}$ & $\frac{n}{2}$ & $\frac{n}{2}$ & 0 & $\frac{n}{4}$ & $\frac{n}{4}$\\
    $D_{\textnormal{square}}$ & $n\sqrt{n}$ & $2\sqrt{n}$ & $\sqrt{n}$
    & $\sqrt{n}$ & $\sqrt{n}$ & $2\sqrt{n}$ & $\sqrt{2\sqrt{n}}$
    & $\sqrt{\sqrt{n}}$ & $\sqrt{\sqrt{n}}$\\
    $D_{\textnormal{triangular}}$ & $\frac{2n\sqrt{2n}}{3}$ & $2\sqrt{2n}$ & $\sqrt{2n}$
    & $\sqrt{2n}$ & 2 & $\frac{4}{3}\sqrt{2n}$ & $\frac{2}{3}\sqrt{n}$
    & $\frac{\sqrt{n}}{3}$ & $\frac{\sqrt{n}}{3}$\\
    \bottomrule
  \end{tabular}
  \begin{tabular}{m{0.08\columnwidth}ccccccccc}
    \toprule
    DAG & $\len$ & $\width$ & $\shape{max}$
    & $\shape{min}$ & $\shape{mean}$ & $\shape{sd}$ & $\shape{1}$ & $\shape{k}$ & $\mass$\\
    \midrule
    $D_{\textnormal{empty}}$ & 1 & $n$ & $n$ & $n$ & $n$ & 0 & $n$ & $n$ & 1\\
    $D_{\textnormal{complete}}$ $D_{\textnormal{chain}}$ & $n$ & 1 & 1 & 1 & 1 & 0 & 1 & 1 & 0\\
    $D_{\textnormal{out-tree}}$ & $\log_2(n)$ & $\frac{n}{2}$ & $\frac{n}{2}$ & 1 & $\frac{n}{\log_2(n)}$ & $\frac{n}{\sqrt{3\log_2(n)}}$ & 1 & $\frac{n}{2}$ & 1\\
    $D_{\textnormal{in-tree}}$ & $\log_2(n)$ & $\frac{n}{2}$ & $\frac{n}{2}$ & 1 & $\frac{n}{\log_2(n)}$ & $\frac{n}{\sqrt{3\log_2(n)}}$ & $\frac{n}{2}$ & 1 & 1\\
    $D_{\textnormal{comb}}$ & $\frac{n}{2}$ & $\frac{n}{2}$ & 2 & 1 & $2$ & $\sqrt{\frac{2}{n}}$ & 1 & 2 & 1\\
    $D_{\textnormal{comb}}^R$ & $\frac{n}{2}$ & $\frac{n}{2}$ & $\frac{n}{2}$ & 1 & $2$& $\sqrt{\frac{n}{2}}$ & $\frac{n}{2}$ & 1 & $\frac1{2}$\\
    $D_{\textnormal{bipartite}}$ & 2 & $\frac{n}{2}$ & $\frac{n}{2}$ & $\frac{n}{2}$ & $\frac{n}{2}$ & 0 & $\frac{n}{2}$ & $\frac{n}{2}$ & 1\\
    $D_{\textnormal{square}}$ & $\sqrt{n}$ & $\sqrt{n}$ & $\sqrt{n}$ & $\sqrt{n}$ & $\sqrt{n}$ & 0 & $\sqrt{n}$ & $\sqrt{n}$ & 1\\
    $D_{\textnormal{triangular}}$ & $\sqrt{2n}$ & $\sqrt{2n}$ & $\sqrt{2n}$ & 1 & $\sqrt{\frac{n}{2}}$ & $\sqrt{\frac{n}{6}}$ & 1 & $\sqrt{2n}$ & 1\\
    \bottomrule
  \end{tabular}
  \caption{\label{tab:prop}Approximate properties of special DAGs
    (negligible terms are discarded for clarity).
    More specifically, each approximate property
    $\textnormal{approx}(n)$ is related to the exact one
    $\textnormal{exact}(n)$ such that
    $\lim_{n\to\infty}\frac{\textnormal{approx}(n)}{\textnormal{exact}(n)}=1$.
    The exact properties are given in
    Appendix~\ref{sec:exact-prop-spec}.}
\end{table}

Table~\ref{tab:prop} illustrates the properties for these special
DAGs.
To discuss them, we analyze the most extreme values for each property.
They are reached with the empty and complete DAGs except for the
maximum standard deviations.
The maximum value for the shape standard deviation is $\frac{n-1}{2}$
(reached with an empty DAG to which a single edge is added).
When considering only transitive reductions (i.e.\ when discarding the
complete DAG), the maximum value for the maximum degrees remains $n$
with either a fork (a single source vertex is the predecessor of all
other vertices) or a join (the reversed fork).
Proposition~\ref{prop:bipartite} states that the maximum number of
edges among all transitive reductions is
$\left\lfloor {\frac{n^2}{4}} \right\rfloor$ (reached with the
bipartite DAG).
As a corollary, the maximum value for the minimum and mean degrees is
$\frac{n}{2}$.
Studying the maximum achievable values for the degree standard
deviations is left to future work.

\begin{proposition}
  \label{prop:bipartite}
  The maximum number of edges among all transitive reductions
  of size $n$ is $\left\lfloor \frac{n^2}{4} \right\rfloor$.
\end{proposition}

\begin{proof}
  Transitive reductions do not contain triangle (i.e.\ clique of size
  three), otherwise there is either a cycle or a redundant edge.
  By Mantel's Theorem\cite{mantel1907problem}, the maximum number of
  edges in a $n$-vertex triangle-free graph is
  $\left\lfloor \frac{n^2}{4} \right\rfloor$.
  This is the case for the complete bipartite DAG because the number
  of edges is $\frac{n^2}{4}=\left\lfloor \frac{n^2}{4} \right\rfloor$
  when $n$ is even and
  $\frac{n^2-1}{4}=\left\lfloor \frac{n^2}{4} \right\rfloor$ when
  $n$ is odd.
\end{proof}

The edge-related properties are considerably affected when considering
the transitive reduction of the complete DAG, i.e.\ the chain.
Except for the standard deviations, all such properties are divided by
$O(n)$.
Considering transitive reductions can thus lead to different
conclusions.
The edge-related properties also highlight the asymmetry of both trees
through the difference between input and output degrees.
Moreover, the density of a DAG appears to be quantified by the
edge-related properties (e.g.\ the complete DAG and last three DAGs).
Small values for the degree standard deviations characterize DAGs in
which every vertex shares a similar structure (e.g.\ the empty DAG,
chain, trees and combs).
The length and shape-based properties show whether the DAG is short
(empty and bipartite DAGs), balanced (the trees, square triangular
DAGs) or long (the complete DAG, chain and combs).
The maximum shape equals the width except for the reversed comb, which
confirms the results shown in Table~\ref{tab:comparison} on the
similarity between the maximum shape and the width.
Finally, large values for the shape standard deviation characterize DAGs for which
the parallelism varies significantly.
This is the case for the trees and triangular DAG\@.

The analysis of these special DAGs provides some insight to select the
relevant properties in the rest of this paper.
Each given DAG possesses 18 properties, to which we add 9 properties
by considering the transitive reduction and 7 properties by
considering the reversal (for a total of 34 properties, $n$ excluded).
We limit the scope of our study to discard some properties for
simplicity.
First, we assume that the generated DAGs are symmetrical and have
similar properties through the reversal operation (which is the case
for all special DAGs except for the trees and combs).
This eliminates the 7 properties on the reversal.
Moreover, the following properties become redundant with
$\degree{}{max}$ and $\degree{}{sd}$: $\degree{in}{max}$,
$\degree{out}{max}$, $\degree{in}{sd}$ and $\degree{out}{sd}$ (which
eliminates 8 additional properties).
Second, we assume that only transitive reductions are meaningful in
the context of scheduling without communication.
This eliminates only 4 other properties because we keep the number of
edges in the initial DAG because it provides meaningful information on
the generation method.
Moreover, we discard the mean degree because it is redundant with $n$
and $m$, and provides little insight.
Similarly, the minimum degree is not kept because it may be
uninformative as it is low for source and sink vertices.
We also discard the width and maximum shape because the mean shape
provides a more global information.
The mass already takes into account the minimum shape, which we
discard.
The last two shape properties ($\shape{1}$ and $\shape{k}$) provides
only local information and are thus not kept.

This leaves 8 properties.
In particular, we measure the following edge-related properties on the
transitive reduction of any DAG\@: the number of edges, maximum degree
and degree standard deviation.
Additionally, we keep the length, the mean shape (even though it is
redundant with $n$ and $\len$, it provides essential information on
the global parallelism of the DAG), the shape standard deviation and
the mass.
The final property is the number of edges in the initial DAG\@.

\begin{figure}
  \centering
  \includegraphics{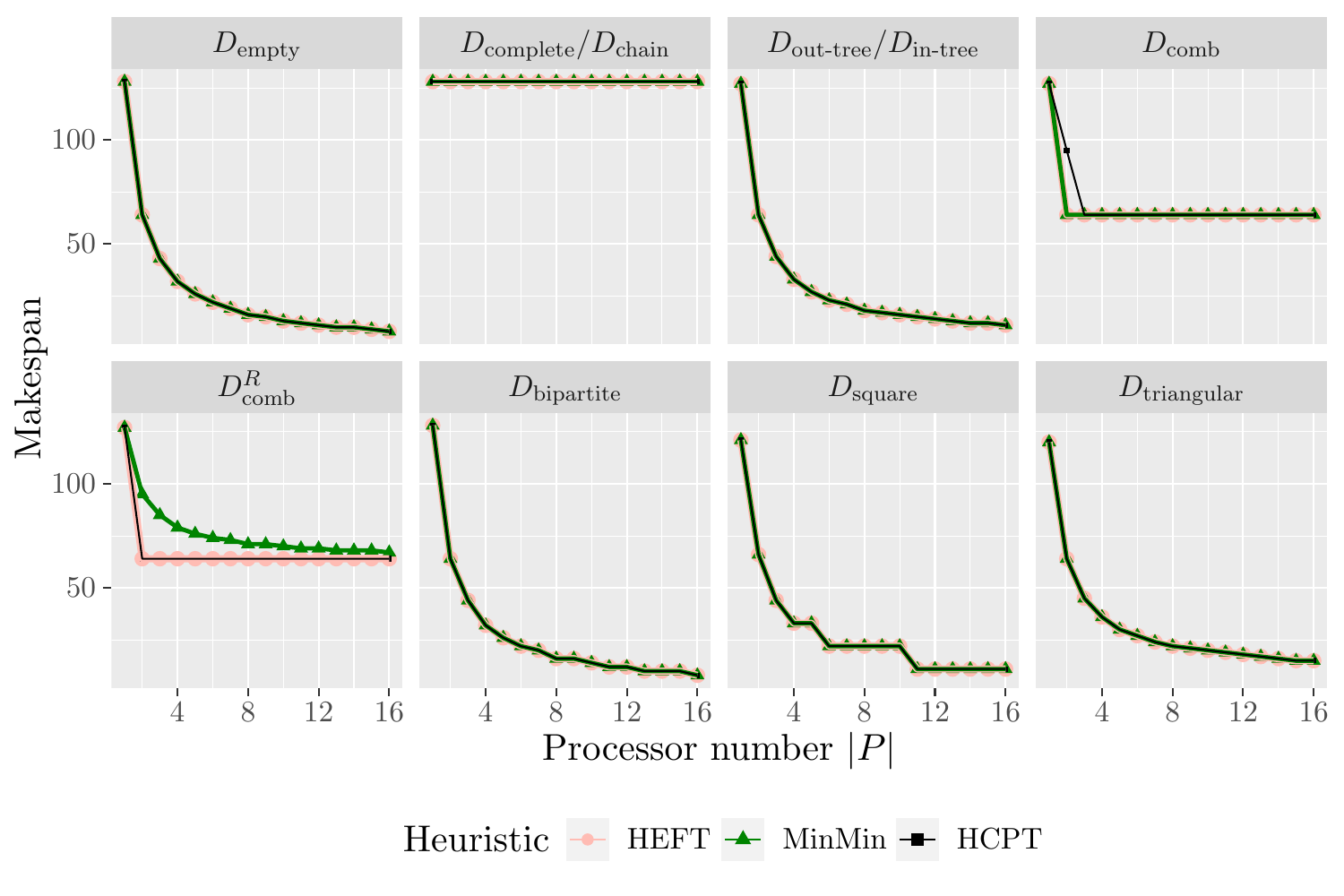}
  \caption{\label{fig:spe_schedule}Makespan obtained with three
    heuristics (described in Section~\ref{sec:eval-sched-algor}) on
    all special DAGs of Section~\ref{sec:analys-spec-dag} for
    \problem.
    The number of vertices is $n=128$ for the empty DAG, complete and
    bipartite DAGs, $n=127$ for the trees and combs, $n=121$ for the
    square DAG, and $n=120$ for the triangular DAG
    ($1+\cdots+15=120$).
  }
\end{figure}

Figure~\ref{fig:spe_schedule} shows the makespan obtained with three
scheduling heuristics with all special DAGs as the number of
processors varies.
HEFT is always optimal because of the regularity of the DAG structures
and because costs are unitary.
This is also the case for the other heuristics most of the time.
A zero mass, for long DAGs such as the complete DAG and chain, leads
to an even easier scheduling problem where the number of processors
has no impact.
This confirms the discussion in Section~\ref{sec:mass-scheduling}
stating that low mass is characteristic of easy instances.
For the other DAGs, increasing the number of processors decreases the
makespan until it reaches 1, 2, 7, 11, 15 and 50 for the empty DAG,
bipartite DAG, trees, square DAG, triangular DAG and combs,
respectively.
Note that the stairs for the square are due to its layered structure.
For the reversed comb, MinMin behaves poorly because this simple
heuristic does not take into account the critical path and fill the
processors with any of the initial source vertices.
Finally, the sub-optimal schedule produced by HCPT for the comb DAG is
because, contrarily to HEFT with its insertion mechanism, this
heuristic does not rely on backfilling and cannot schedule a task
before any other already scheduled tasks.

\section{Analysis of Existing Generation Methods}\label{sec:existinggeneration}

This section covers and analyzes existing generation methods: the
classic \erdos algorithm; a uniform random generation method via a
recursive approach; a poset-based method; and, an ad-hoc method
frequently used in the scheduling literature.

\subsection{Random Generation of Triangular Matrices}
\label{sec:rand-gener-triang}

This approach is based on the \erdos algorithm\cite{erdos1959a} with
parameter $p$ (noted $G(n,p)$ in\cite{bollobas2001random}): an
upper-triangular adjacency matrix is randomly generated.
For each pair of vertices $(i,j)$, with $i<j$, there is an edge from
$i$ to $j$ with an independent probability $p$.
The expected number of edges is therefore $p\frac{n(n-1)}{2}$.

The approach is not uniform (nor uniform up to isomorphism).
For instance, a generator that is uniform up to isomorphism picks up
the empty DAG with probability $1/6$ (see Table~\ref{tab:dag3}).
Moreover, a random generator that is uniform over all the DAGs (see
Section~\ref{sec:unif-rand-gener} for the distinction) generates the
empty DAG with probability $1/25$.
With $p=0.5$, the \erdos algorithm generates the DAG with no edges
with probability $1/8$.

\begin{figure}
  \centering
  \includegraphics{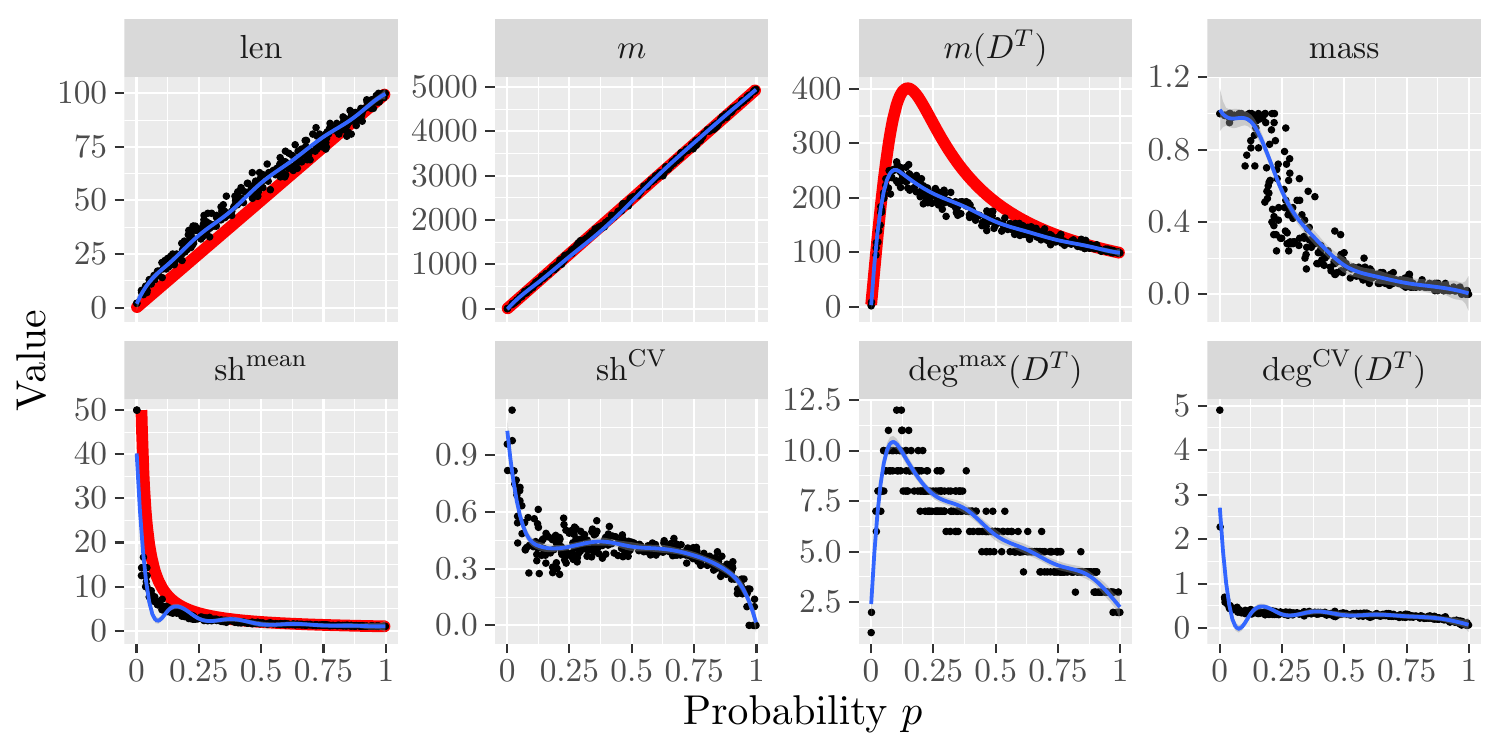}
  \caption{\label{fig:erdos_proba} Properties of 300 DAGs of size
    $n=100$ generated by the \erdos algorithm with probability $p$
    uniformly drawn between 0 and 1.
    The smoothed line is obtained with a linear regression using a
    polynomial spline with 10 degrees of freedom.
    The degree CV (Coefficient of Variation) is the ratio of the mean
    degree to the degree standard deviation.
    Red lines correspond to formal results for the length and mean
    shape (Proposition~\ref{prop:ERmean}), the number of edges, and
    the number of edges in the transitive reduction
    (Proposition~\ref{prop:ERET}).
  }
\end{figure}

Figures~\ref{fig:erdos_proba} and~\ref{fig:erdos_node} show the effect
of both parameters, probability $p$ and size $n$, on the properties of
the generated DAGs.
For readability of both figures, each standard deviation is replaced
by a CV (Coefficient of Variation), which is the ratio of the standard
deviation to the mean.
The most evident effect on both figures is that the number of edges
$m$ increases linearly as $p$ increases and quadratically as $n$
increases, which is a direct consequence of the algorithm and the
expected number of edges.
Similarly, but with more variation, the length also increases as
either parameter increases.
This effect also concerns the mean shape because
$\shape{mean}=\frac{n}{\len}$ (for instance, the length is close to 20
when $p=0.125$, whereas the mean shape is close to 5).
Therefore, on Figure~\ref{fig:erdos_proba}, the mean shape decreases
as the inverse function of the probability $p$ because the length
increases quasi-linearly with $p$.
This effect is consistent with Proposition~\ref{prop:ERmean} in
Appendix~\ref{sec:annexe:probaER}, which suggests that the expected
mean shape is no greater than $\frac1{p}$.

A more remarkable effect can be seen for the number of edges in the
transitive reduction $m(D^T)$.
This property shows that after a maximum around $p=0.10$, adding more
edges with higher probabilities leads to redundant dependencies and
simplifies the structure of the DAG by making it longer.
The same observation can be done with $\degree{}{max}(D^T)$.
This is consistent with the fact that the algorithm generates the
empty DAG when $p=0$ and the complete DAG when $p=1$.
Proposition~\ref{prop:ERET} in Appendix~\ref{sec:annexe:probaER} also
confirms this effect.

We rely on this apparent threshold around $p=10\%$ to characterize
three probability intervals: below 5\%, between 5\% and 15\%, and
above 15\%.
DAGs generated with a probability in the first interval are almost
empty (hence a length lower than 10 and a mean shape higher than 10)
with few vertices having some edges and many with no edges (hence the
high degree standard deviation).
For these DAGs, most edges are not redundant.
Given the high shape standard deviation, many tasks must be available
at first.
As mentioned in Section~\ref{sec:scheduling}, these DAGs lead to a
simplistic scheduling process that consists in starting each task on a
critical path as soon as possible and then distributing a large number
of independent tasks.
Analogously, DAGs generated with probabilities $p$ greater than 15\%
contain many edges that simplify the DAG structure by increasing the
length and thus reducing the mean shape (recall that with a small
width, the problem is easy, see Section~\ref{sec:scheduling}).
At the same time, the mass decreases continuously, allowing the
problem to be divided into smaller problems.
In particular, for probability $p$ greater than 90\%, DAGs are close
to the chain, which is trivial to schedule.
Therefore, most interesting DAGs are generated with probabilities
between 5\% and 15\%.

\begin{figure}
  \centering
  \includegraphics{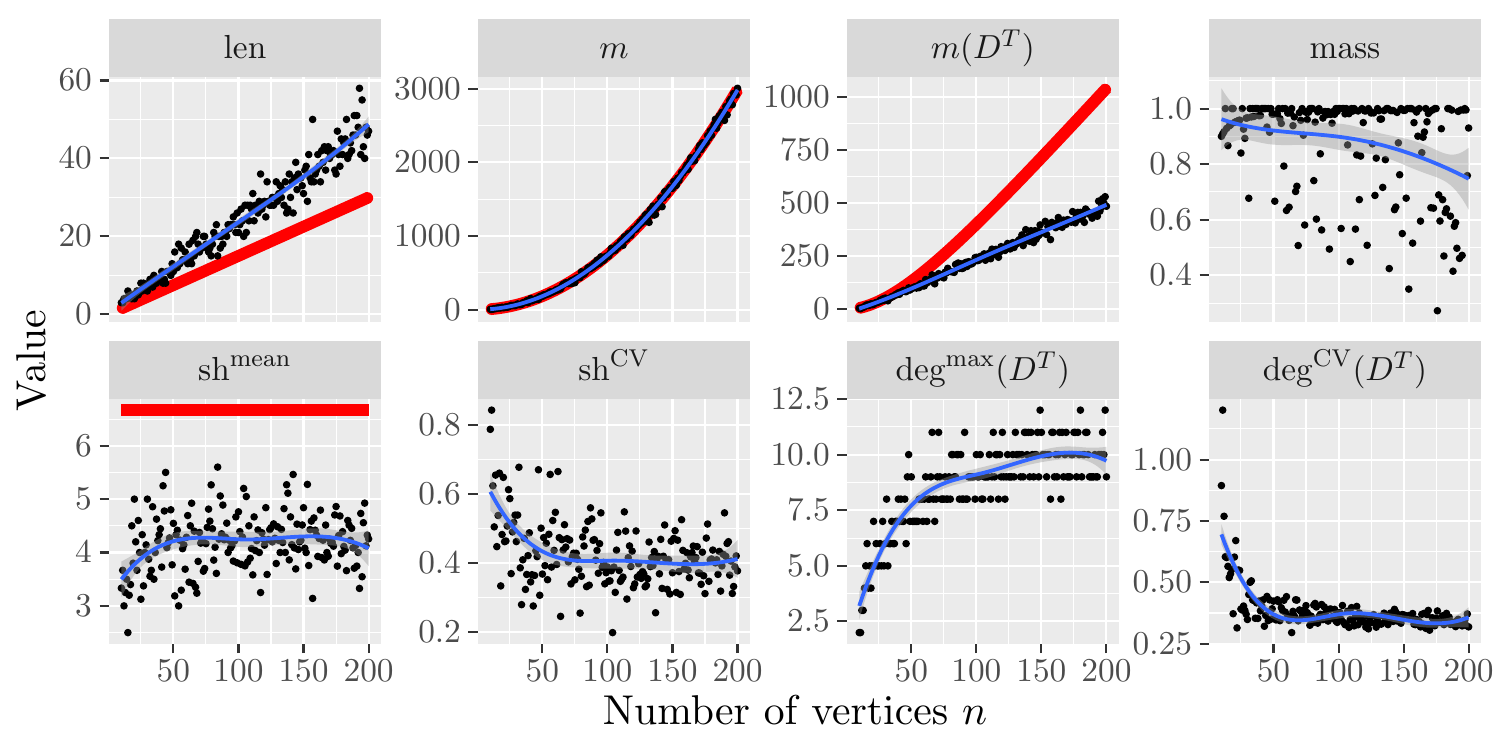}
  \caption{\label{fig:erdos_node} Properties of 191 DAGs generated by
    the \erdos algorithm with probability $p=0.15$ and for each size
    $n$ between 10 and 200.
    The smoothed line is obtained with a linear regression using a
    polynomial spline with 4 degrees of freedom.
    Red lines correspond to formal results for the length and mean
    shape (Proposition~\ref{prop:ERmean}), the number of edges, and
    the number of edges in the transitive reduction
    (Proposition~\ref{prop:ERET}).
  }
\end{figure}

As shown on Figure~\ref{fig:erdos_node}, the size of the DAG $n$ has a
simpler effect on the number of edges in the transitive reduction
$m(D^T)$ than the probability $p$: $m(D^T)$ increases linearly with
$n$ (see Proposition~\ref{prop:ERET}).
Moreover, the length increases with $n$ as the shape mean remains
constant (see Proposition~\ref{prop:ERmean}).
As a consequence, the mass decreases with $n$ because the probability
to obtain the value 1 increases in a vector with constant mean but
increasing size.
It is thus advisable to lower the probability with large sizes to
maintain a constant mass.

The analysis of the \erdos algorithm provides some insight on the
desirable characteristics for the purpose of comparing scheduling
heuristics.
The effect of probability $p$ illustrates the compromise between the
length and mean shape to avoid simplistic instances that are easily
tackled (see Section~\ref{sec:scheduling}).
Moreover, the maximum number of edges in the transitive reduction
$m(D^T)$ is around $\frac5{2}n$ in both figures.
However, we know that reaching $\frac{n^2}{4}$ is possible
(Proposition~\ref{prop:bipartite}) and layer-by-layer DAGs (square and
triangular) are in $O(n^{\frac{3}{2}})$.
Therefore, the \erdos algorithm fails to generate DAGs with such large
$m(D^T)$.

\subsection{Uniform Random Generation}
\label{sec:recursive}

There are two main ways to provide a uniform random generator to
uniformly generate elements of $\Dn$ (uniform over all labelled DAGs,
see Section~\ref{sec:unif-rand-gener}).
The first one consists in using a classical recursive/counting
approach\cite{roblab}.
This counting approach relies on recursively counting the number of
DAGs with a given number of source vertices, that is vertices with no
in-going edges.
See\cite[Section~4]{DBLP:journals/sac/KuipersM15} for a complete
algorithm that uniformly generates random DAGs with this approach.
The second one relies on MCMC
approaches\cite{DBLP:journals/endm/MelanconDB01,DBLP:conf/sbia/IdeC02,DBLP:journals/ipl/MelanconP04}.
We describe below the recursive approach.

Let $a_n=|\Dn|$, $a_{n,s}$ be the number of DAGs of $\Dn$ having
exactly $s$ source vertices ($\shape{1}=s$).
It is proved in\cite{roblab} that:
$$a_n=\sum_{k=1}^na_{n,k}\quad \text{and}\quad
a_{n,k}={\binom{n}{k}}b_{n,k} \quad \text{with}\quad
b_{n,k}=\sum_{s=1}^{n-k}(2^k -1)^s2^{k(n-k-s)}a_{n-k,s}.$$

First, we compute all values $a_i$ and $a_{i,k}$ for $1\leq i\leq n$
and $1\leq k\leq i$ with the initial conditions $a_{i,i}=1$ for
$1\leq i\leq n$.
Next, a shape is generated using Algorithm~\ref{algo:shaperec}, where
$\oplus$ is the concatenation of vectors.

\begin{algorithm}
  \SetAlgoLined

  \KwData{$n$, $a_{i}$, $a_{i,k}$ for $1\leq i\leq n$ and
    $1\leq k\leq i$.}
  \KwResult{A shape with $n$ elements.}

  Randomly generate $s\in \{1,n\}$ with distribution
  $\Prob(s=j)=\frac{a_{n,j}}{a_n}$\;
  \KwRet{$[s]\oplus\textnormal{RandomShape}(n-s)$\;}
\caption{$\textnormal{RandomShape}(n)$}\label{algo:shaperec}
\end{algorithm}

Finally, Algorithm~\ref{algo:shapeDAG} builds the final DAG by adding
random edges.

\begin{algorithm}
  \SetAlgoLined

  \KwData{$[s_1,\ldots,s_k]$ a shape with $n$ elements.}
  \KwResult{A DAG with $n$ vertices.}

  \For{$i\in [1,\ldots,k]$}{
    \For{$j\in [1,\ldots,s_i]$}{
      Generate a vertex $v$ with level $i$\;
      \If{$i>1$}{
        Connect a random vertex from level $i-1$ to vertex $v$\;
        Connect any other vertex from previous levels to vertex $v$
        with probability 0.5\;
      }
    }
  }
  \KwRet{the resulting DAG\;}
\caption{$\textnormal{ShapeToDAG}([s_1,\ldots,s_k])$}\label{algo:shapeDAG}
\end{algorithm}

\begin{figure}
  \centering
  \includegraphics{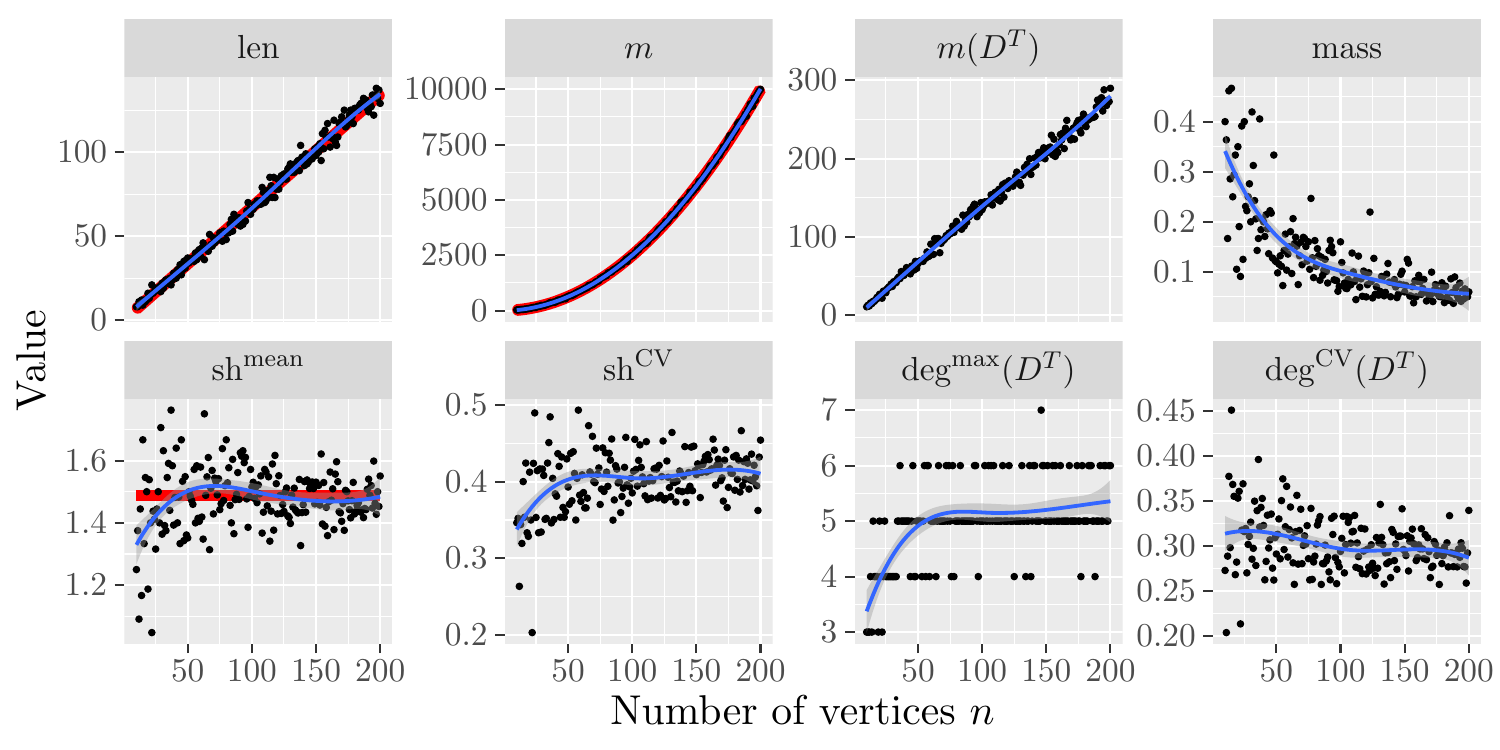}
  \caption{\label{fig:recursive} Properties of 191 DAGs generated by
    the recursive algorithm for each size $n$ between 10 and 200.
    The smoothed line is obtained with a linear regression using a
    polynomial spline with 4 degrees of freedom.
    Red lines correspond to formal results for the length and mean
    shape, and the number of edges (the bound from
    Theorem~\ref{theorem:mass} is discarded because it is too far).
  }
\end{figure}

Figure~\ref{fig:recursive} depicts the effect of the number of
vertices on the selected DAG properties.
Three effects are noteworthy: the length closely follows the function
$\frac{3n}{2}$, the number of edges $m$ is almost indistinguishable
from the function $\frac{n^2}{4}$ and the number of edges in the
transitive reduction $m(D^T)$ closely follows $1.4n$.
The first effect is consistent with a theoretical result stating that
the expected number of source vertices $\shape{1}$ in a uniform DAG is
asymptotically 1.488 as $n\to\infty$\cite{liskovetsmaixmal}.
This implies that the expected value for each shape element is close
to this value by construction of the shape.
Proposition~\ref{prop:expectshape} in Appendix~\ref{sec:annexe:proba}
confirms this expectation is no larger than 2.25, which makes the DAG
an easy instance for scheduling problems (see
Section~\ref{sec:scheduling}).
For the second effect, we know that the average number of edges in a
uniform DAG is indeed $\frac{n^2}{4}$\cite[Theorem
2]{DBLP:journals/endm/MelanconDB01}.
Despite the large amount of studies dedicated to formally analyzing
uniform random DAGs, to the best of our knowledge, the last effect has
not been formally considered.
We finally observe that the mass decreases as the size $n$ increases.
This is confirmed by the following result, proved in
Appendix~\ref{sec:annexe:proba}:

\begin{theorem}\label{theorem:mass}
  Let $D$ be a DAG uniformly and randomly generated among the labeled
  DAGs with $n$ vertices.
  One has $\Prob(\absmass(D)\geq \log^4(n))\to 0$ when
  $n\to+\infty$.
\end{theorem}

Therefore, the mass converges to zero as the size $n$ tends to
infinity.
As shown in Section~\ref{sec:mass-scheduling}, such instances can be
decomposed into independent problems and efficiently solved with a
brute force strategy.
This leads to a sub-exponential generic time complexity with uniform
instances.

To obtain a similar average number of edges $m$ with the \erdos
algorithm, we must choose a probability $p=0.5$.
We can compare both methods by considering $p=0.5$ and $n=100$ on
Figures~\ref{fig:erdos_proba} and~\ref{fig:recursive}, respectively.
We observe that DAGs generated by both methods share similar
properties.
This leads to similar conclusions as in
Section~\ref{sec:rand-gener-triang}.

\subsection{Random Orders}
\label{sec:random-orders}

The random orders method derives a DAG from randomly generated
orders\cite{winkler1985random}.
The first step consists in building $K$ random permutations of $n$
vertices.
Each of these permutations represents a total order on the vertices,
which is also a complete DAG with a random labeling.
Intersecting these complete DAGs by keeping an edge iff it appears in
all DAGs with the same direction leads to the final DAG\@.
This is a variant of the algorithm presented in\cite{cordeiro2010a}
where the transitive reduction in the last step is not performed
because we already measure the properties on the transitive reduction.

\begin{figure}
  \centering
  \includegraphics{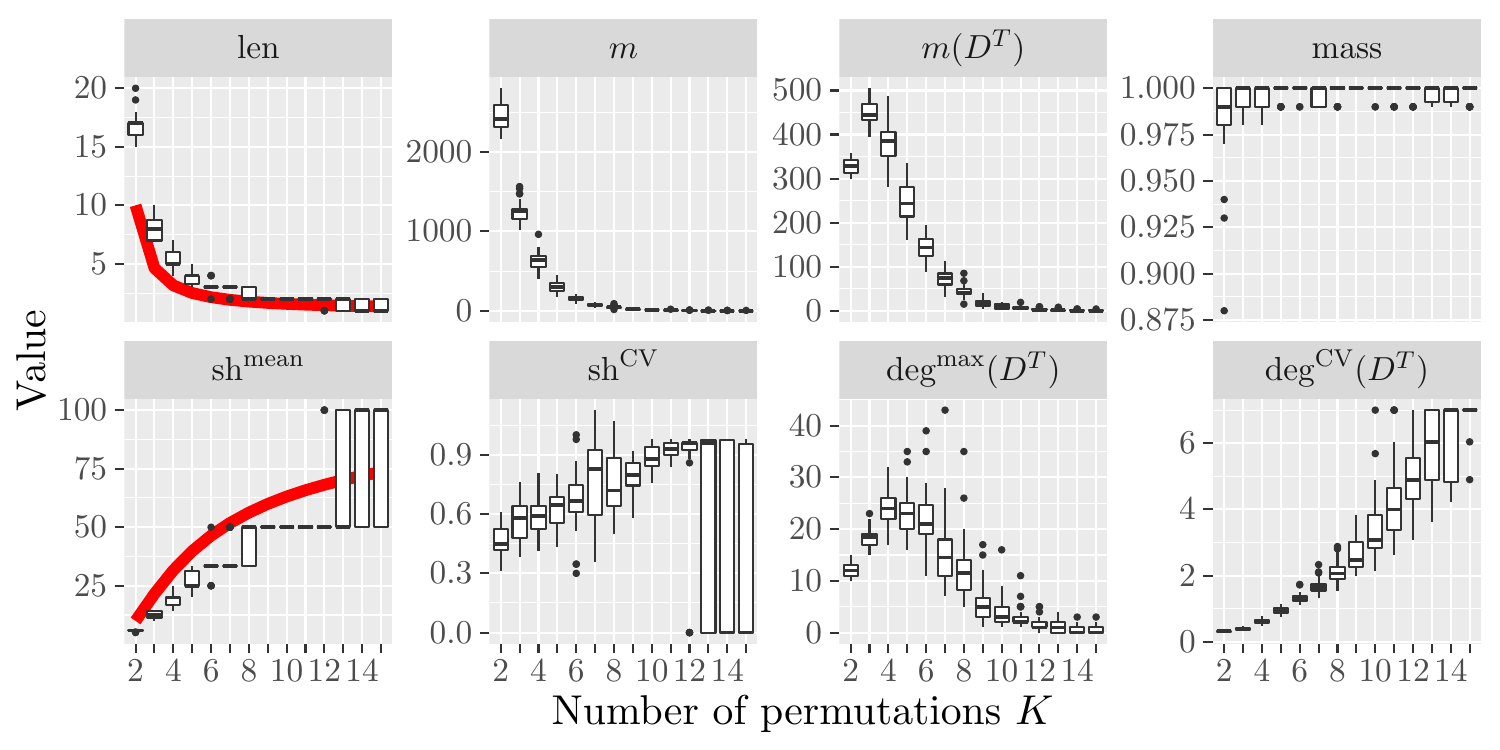}
  \caption{\label{fig:poset_perm} Properties of 420 DAGs of size
    $n=100$ generated by the random orders algorithm for each number
    of permutations $K$ between 2 and 15 (30 DAGs per boxplot).
    Red lines correspond to formal results for the length and mean
    shape.
  }
\end{figure}

Figure~\ref{fig:poset_perm} shows the effect of the number of
permutations $K$ on the DAG properties with boxplots\footnote{Each
  boxplot consists of a bold line for the median, a box for the
  quartiles, whiskers that extend at most to 1.5 times the
  interquartile range from the box and additional points for
  outliers.}.
The extreme cases $K=1$ and $K\to\infty$ are discarded from the figure
for clarity.
They correspond to the chain and the empty DAG, respectively.
Recall that for the chain, $m(D^T)\approx\len=100$,
$m\approx\numprint{5000}$, $\shape{mean}=\degree{}{max}(D^T)=1$ and
the CVs and mass are zero.
Similarly, for the empty DAG, $\len=\mass=1$, the mean shape is 100
and all the other properties are zero.

The number of permutations quickly constrains the length.
For instance, the length is already between 15 and 20 when $K=2$ and
at most 5 when $K\ge5$.
A formal analysis suggests that the length is almost surely in
$O(n^{1/K})$\cite[Theorem 3]{winkler1985random}, which is consistent
with our observation.
The number of edges and the maximum degree in the transitive reduction
reach larger values than with previous approaches for any size $n$
(twice larger than with the \erdos algorithm).
Moreover, the mass is always close to one for $K>1$.
Some specific values can finally be explained.
First, the maximum value for $\degree{}{CV}(D^T)$ is exactly 7 and
corresponds to DAGs of size $n=100$ with a single edge (2 vertices
have degree 1 and 98 others have 0).
Also, the shape CV is at most 0.98 when the length is 2 (which
frequently when $K\ge10$).
This CV corresponds to a shape with values 99 and 1.

\begin{figure}
  \centering
  \includegraphics{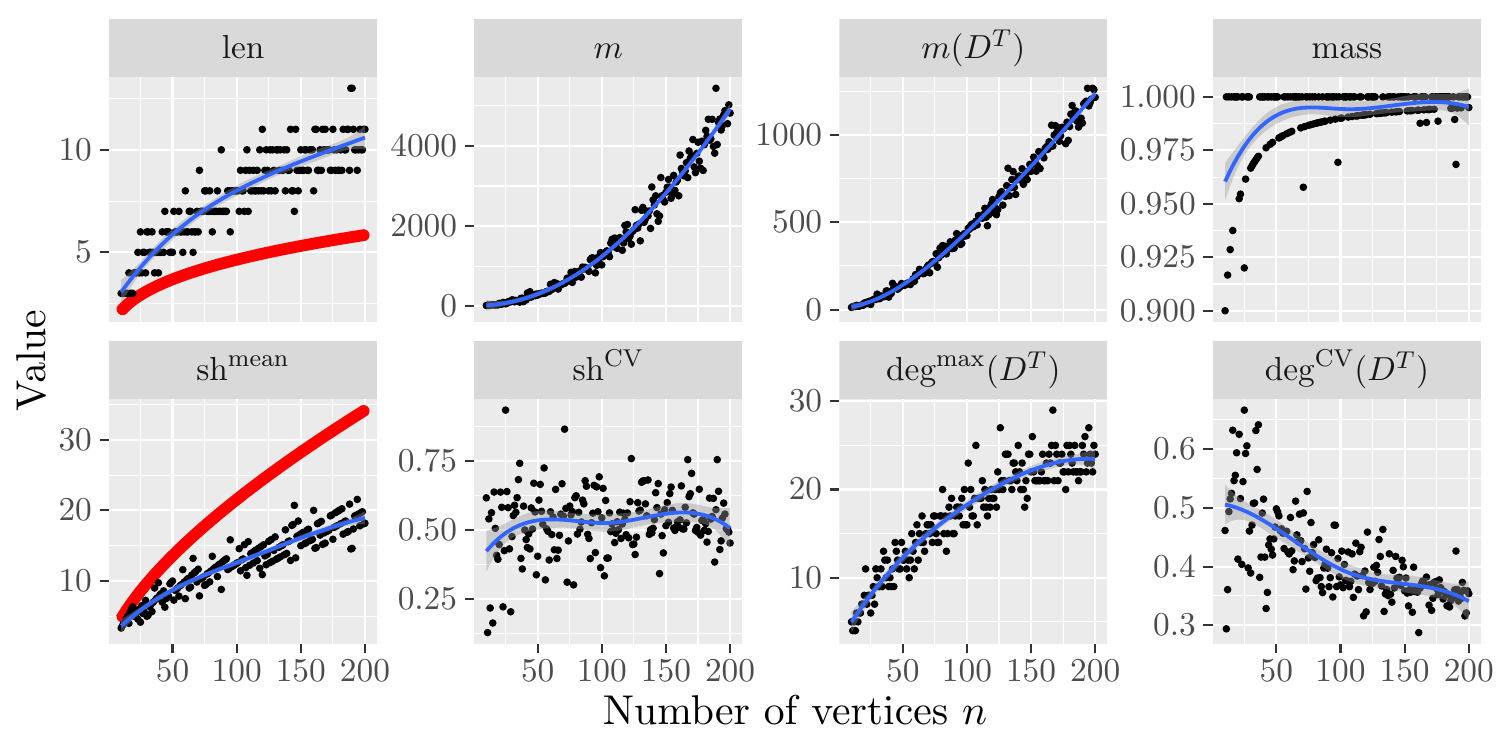}
  \caption{\label{fig:poset_node} Properties of 191 DAGs generated by
    the random orders algorithm with $K=3$ permutations and for each
    size $n$ between 10 and 200.
    The smoothed line is obtained with a linear regression using a
    polynomial spline with 4 degrees of freedom.
    Red lines correspond to formal results for the length and mean
    shape.
  }
\end{figure}

Figure~\ref{fig:poset_node} shows the effect of the number of vertices
$n$ for a fixed number of permutations $K$.
We selected $K=3$ to have the maximum number of edges in the
transitive reduction.
The sublinear relation between the length and size $n$ is again
consistent with the previously cited result (i.e.\ $O(n^{1/K})$).
Even though $K=3$ is small, the length is already low, leading to line
patterns for both the length and the mean shape.
Note that the mass is frequently either 1 or almost 1 (i.e.,
$1-\frac1{k}$), which corresponds to cases where only the last value
of the shape $\shape{k}$ is one.

The random orders method can generate denser DAGs than \erdos or
uniform DAGs without the mass issue, but with difficult control over
the compromise between the length and the mean shape.

\subsection{Layer-by-Layer}

Many variants of the layer-by-layer principle have been used
throughout the literature to assess scheduling algorithms and are
covered in Section~\ref{sec:layer-layer-method}.
This section analyzes the effect of three parameters (size $n$, number
of layers $k$ and connectivity probability $p$) using the following
variant inspired from\cite{cordeiro2010a,gupta2017a}.
First, $k$ vertices are affected to distinct layers to prevent any
empty layer.
Then, the remaining $n-k$ vertices are distributed to the layers using
a balls into bins approach (i.e.\ a uniformly random layer is selected
for each vertex).
For each vertex not in the first layer, a random parent is selected
among the vertices from the previous layer to ensure that the layer of
any vertex equals its depth (similar to\cite{dutot2009a,gupta2017a}
and the recursive method in Section~\ref{sec:recursive}).
Finally, random edges are added by connecting any pair of vertices
from distinct layers from top to bottom with probability $p$.

This variant departs from\cite{cordeiro2010a,gupta2017a} to ensure
generated DAGs have a length equal to $k$ and mean shape equal to
$n/k$.
Moreover, with some parameter values, this method produces some of the
special DAGs covered in Section~\ref{sec:analys-spec-dag}.
It generates the empty DAG when $k=1$, whereas it generates the
complete DAG with $k=n$ and $p=1$.
To interpret the number of edges depicted in
Figures~\ref{fig:layer_proba} to~\ref{fig:layer_node}, we study the
case (called \emph{regular}) when all layers have the same size $n/k$,
which constitutes an approximation of the DAGs generated by the
layer-by-layer variant studied in this section.
When $p=1$, the DAG is the bipartite one for $k=2$ and the square one
for $k=\sqrt{n}$.
In such DAGs and when $n$ is a multiple of $k$, the expected number of
edges is
\begin{equation}\label{eq:edges_regular}
  \Expect(m)=n\left(1-\frac{1}{k}\right)\left(p\left(\frac{n}{2}-1\right)+1\right)
\end{equation}
and the expected number of edges in the transitive reduction is
\begin{equation}\label{eq:edges_red_regular}
  \Expect(m(D^T))\ge p(k-1)\left(\frac{n}{k}\right)^2+(1-p)n\left(1-\frac{1}{k}\right).
\end{equation}

\begin{figure}
  \centering
  \includegraphics{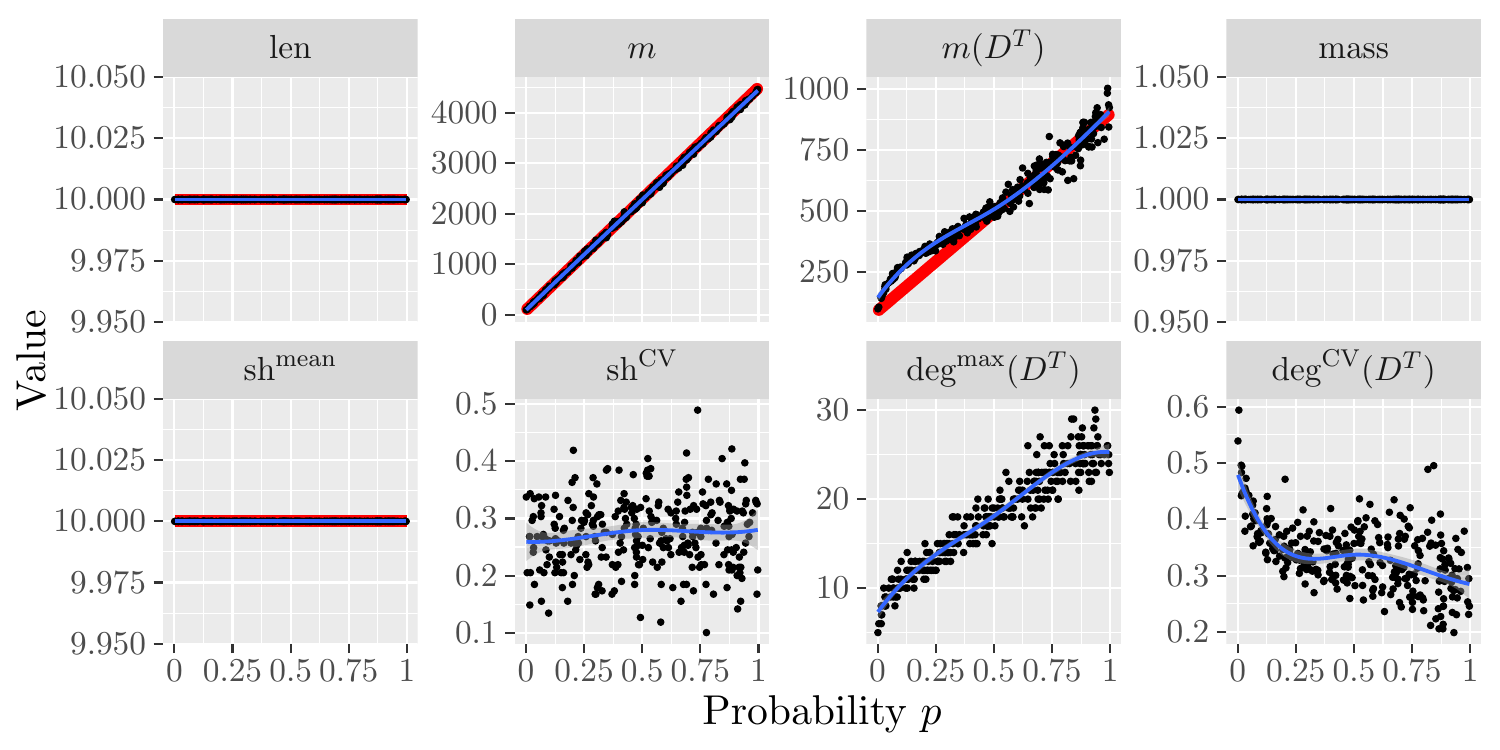}
  \caption{\label{fig:layer_proba} Properties of 300 DAGs of size
    $n=100$ generated by the layer-by-layer algorithm with $k=10$
    layers and probability $p$ uniformly drawn between 0 and 1.
    The smoothed line is obtained with a linear regression using a
    polynomial spline with 4 degrees of freedom.
    Red lines correspond to formal results for the length and mean
    shape, the number of edges (Equation~\ref{eq:edges_regular}), the
    number of edges in the transitive reduction
    (Equation~\ref{eq:edges_red_regular}), and the mass.
  }
\end{figure}

Figure~\ref{fig:layer_proba} shows the effect of the probability $p$.
The analysis for \emph{regular} layer-by-layer DAGs closely
approximates the results.
The number of edges $m$ is predicted to increase linearly from 90 to
\numprint{4500} (Equation~\ref{eq:edges_regular}), while this quantity
in the transitive reduction $m(D^T)$ is expected to increase from 90
to 900 (Equation~\ref{eq:edges_red_regular}).
Remark that this last property undergoes a steeper increase for
probability $p<0.1$ than for larger $p$.
With many edges ($p>0.1$), adding a new one is likely to result into
the introduction of redundant edges, which is not the case for
$p<0.1$.
More generally, the layered structure ensures a steady increase of
$m(D^T)$ as the probability $p$ increases because any edge between two
consecutive layers cannot become redundant through the insertion of
any edge.
The mass is always close to one because the probability to have a
layer with one vertex is close to zero with $k=10$ layers.

\begin{figure}
  \centering
  \includegraphics{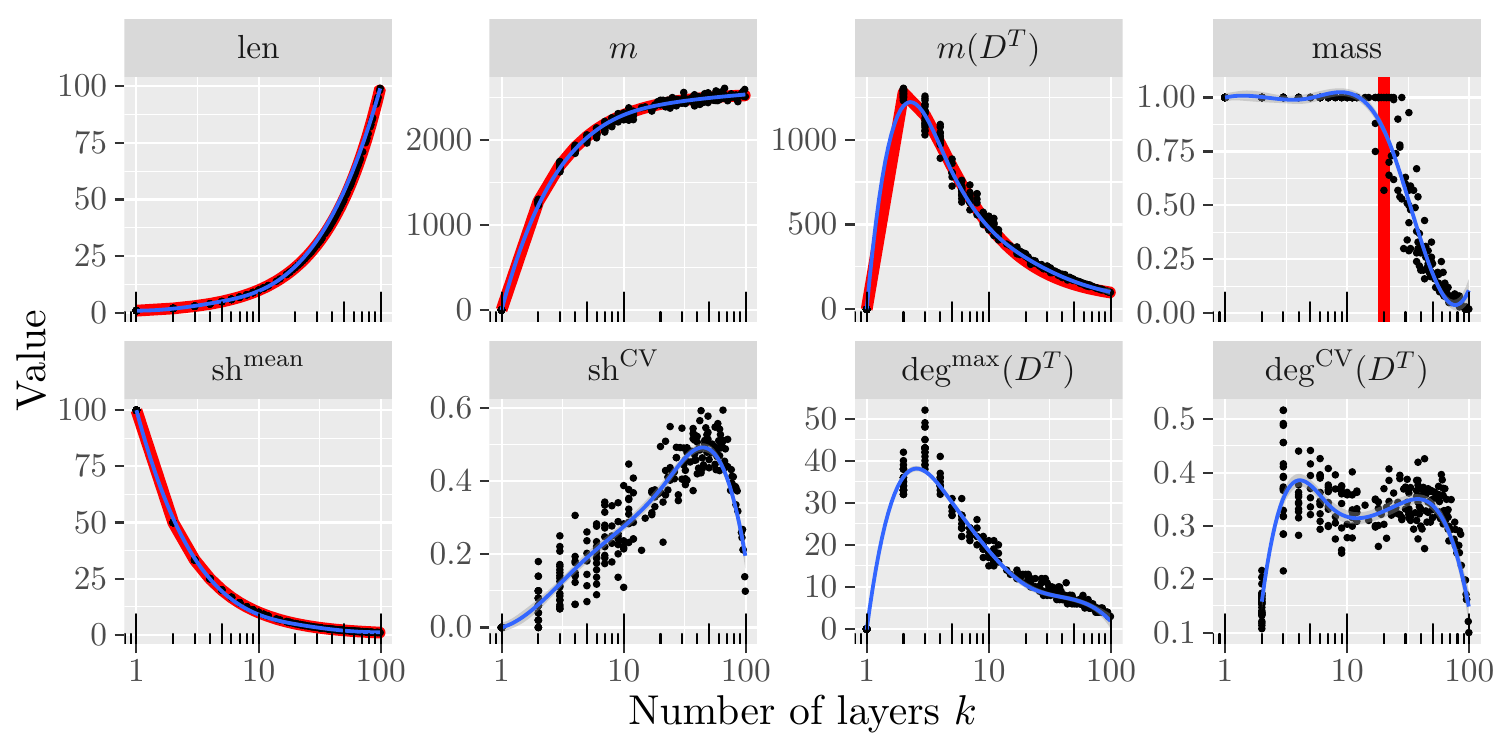}
  \caption{\label{fig:layer_layer} Properties of 300 DAGs of size
    $n=100$ generated by the layer-by-layer algorithm with probability
    $p=0.5$ and a number of layers $k$ randomly drawn between 1 and
    100 ($k=\lfloor {e^{\mathcal{U}(\log(1),\log(101))}} \rfloor$
    where $\mathcal{U}(a,b)$ is a uniform distribution between $a$ and
    $b$).
    The smoothed line is obtained with a linear regression using a
    polynomial spline with 5 degrees of freedom.
    Red lines correspond to formal results for the length and mean
    shape, the number of edges (Equation~\ref{eq:edges_regular}), the
    number of edges in the transitive reduction
    (Equation~\ref{eq:edges_red_regular}), and the mass.
  }
\end{figure}

Figure~\ref{fig:layer_layer} represents the effect of the number of
layers $k$.
With \emph{regular} layer-by-layer DAGs, the expected number of edges
$\Expect(m)$ goes from 0 to \numprint{2524.5} for $k=1$ to 100
(Equation~\ref{eq:edges_regular}), which is close to the results with
our layer-by-layer variant.
The increase is steep because it is already \numprint{2295} for
$k=10$, which is consistent with Figure~\ref{fig:layer_layer}.
The number of edges in the transitive reduction $\Expect(m(D^T))$
decreases from an expected value of \numprint{1275} to 99 as the
number of layers goes from $k=2$ to 100
(Equation~\ref{eq:edges_red_regular}).
The expected value for $k=10$ is 495 and is consistent with both
Figures~\ref{fig:layer_proba} and~\ref{fig:layer_layer}.
Finally, the mass is unitary when there are at least two balls in each
bin.
Since there is initially one ball per bin, this occurs when there is
at least one of the $n-k$ additional balls in each of the $k$ bin.
To compute if there are enough additional balls to have a unitary mass
with probability greater than $0.5$, we can use a bound for the coupon
collector problem\cite[Proposition 2.4]{levin2017markov}.
This occurs when $\left\lceil k\log(2k) \right\rceil+k<n$, which is
the case for $k\le20$ with $n=100$.
This is consistent with Figure~\ref{fig:layer_layer} where the mass
becomes non-unitary around this value.

\begin{figure}
  \centering
  \includegraphics{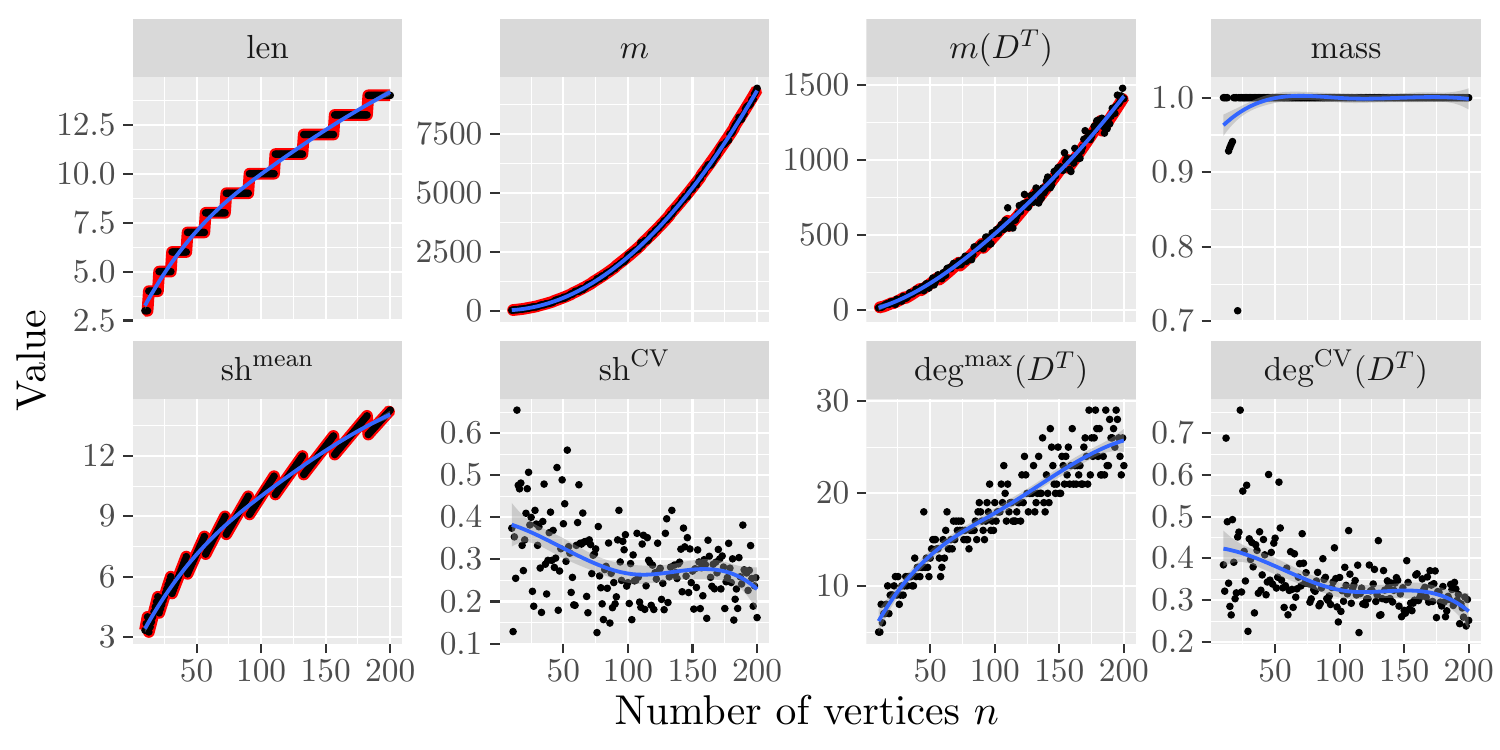}
  \caption{\label{fig:layer_node} Properties of 191 DAGs generated by
    the layer-by-layer algorithm with probability $p=0.5$,
    $k=\sqrt{n}$ (rounded to closest integer) layers and for each size
    $n$ between 10 and 200.
    The smoothed line is obtained with a linear regression using a
    polynomial spline with 4 degrees of freedom.
    Red lines correspond to formal results for the length and mean
    shape, the number of edges (Equation~\ref{eq:edges_regular}), the
    number of edges in the transitive reduction
    (Equation~\ref{eq:edges_red_regular}), and the mass.
  }
\end{figure}

When varying the number of vertices $n$, we expect the number of edges
$m$ to increase quadratically from 20 to around \numprint{9380}
(Equation~\ref{eq:edges_regular}), which is consistent with the
results on Figure~\ref{fig:layer_node}.
Similarly, the number of edges in the transitive reduction $m(D^T)$ is
expected to increase quadratically from around 14.4 to around
\numprint{1420} (Equation~\ref{eq:edges_red_regular}).

In Figures~\ref{fig:layer_proba} to~\ref{fig:layer_node}, the length
and mean shape show stable behavior consistent with our expectation.
In all figures, the shape CV can formally be analyzed using the balls
into bins model and we refer the interested reader to the specialized
literature\cite{kolchin1978random}.
Finally, in the transitive reduction, the maximum degree
$\degree{}{max}(D^T)$ has a similar trend as the number of edges
$m(D^T)$.

To avoid non-unitary mass, the layer-by-layer method can be adapted to
ensure that each layer has two vertices initially.
For instance, we can rely on a uniform distribution between two and a
maximum value, or on a balls into bins approach with two balls per bin
initially.
It is also possible to use the method described in\cite[Section
III]{canon2018markov} to have a uniform distribution of the vertices
in the layers over all possible distributions and with a constraint on
the minimum value.

\section{Evaluation on Scheduling Algorithms}
\label{sec:eval-sched-algor}

Generating random task graphs allows the assessment of existing
scheduling algorithms in different contexts.
Numerous heuristics have been proposed for the problem denoted
\problem (homogeneous tasks and processors, see
Section~\ref{sec:scheduling}) or generalization of this problem.
Such heuristics rely on different principles.
Some simple strategies, like MinMin, execute available tasks on the
processors that minimize completion time without considering
precedence constraints.
In contrast, many heuristics sort tasks by criticality and schedule
them with the Earliest Finish Time (EFT) policy (e.g.\ HEFT and HCPT).
Finally, other principles may be also used: migration for
BSA\cite{kwok2000a}, clustering for DSC\cite{yang1994dsc}, etc.
We focus on the impact of generation methods on the performance of a
selection of three heuristics for this problem: MinMin, HEFT and
HCPT\@.

HEFT\cite{topcuoglu2002a} (Heterogeneous Earliest Finish Time) first
computes the upward rank of each task, which can be seen as a reversal
depth (depth in the reversal DAG).
It then consider tasks by decreasing order of their upward ranks and
schedules them with the EFT policy.
Backfilling is performed following an insertion policy that tries to
insert a task at the earliest idle time between two already scheduled
tasks on a processor if the slot is large enough to accommodate it.
The time complexity of this approach is dominated by the insertion
policy in $O(n^2)$.
Numerous heuristics are equivalent to HEFT when tasks and processors
are homogeneous: PEFT\cite{arabnejad2014list}, HLEFT\cite{adam1974a},
HBMCT\cite{sakellariou2004hybrid}.

HCPT\cite{hagras2003simple} (Heterogeneous Critical Parent Trees)
starts by considering any task on a critical path by decreasing order
of their depth.
The objective is to prioritize the ancestors of such tasks and in
particular when their depth is large.
This process generates a priority list of tasks that are then
scheduling with the EFT policy.
The time complexity is $O(m+n\log(n)+n|P|)$ where $|P|$ is the number
of processors.

Finally, MinMin\cite[Algorithm
D]{ibarra1977}\cite[minmin]{freund1998a} considers all available tasks
any time a processor becomes idle and schedules any task on any
available processor.
With homogeneous tasks and processors, this algorithm is equivalent to
MaxMin\cite[Algorithm E]{ibarra1977}\cite[maxmin]{freund1998a}.
The time complexity is $O(m)$.

\begin{figure}
  \centering
  \includegraphics{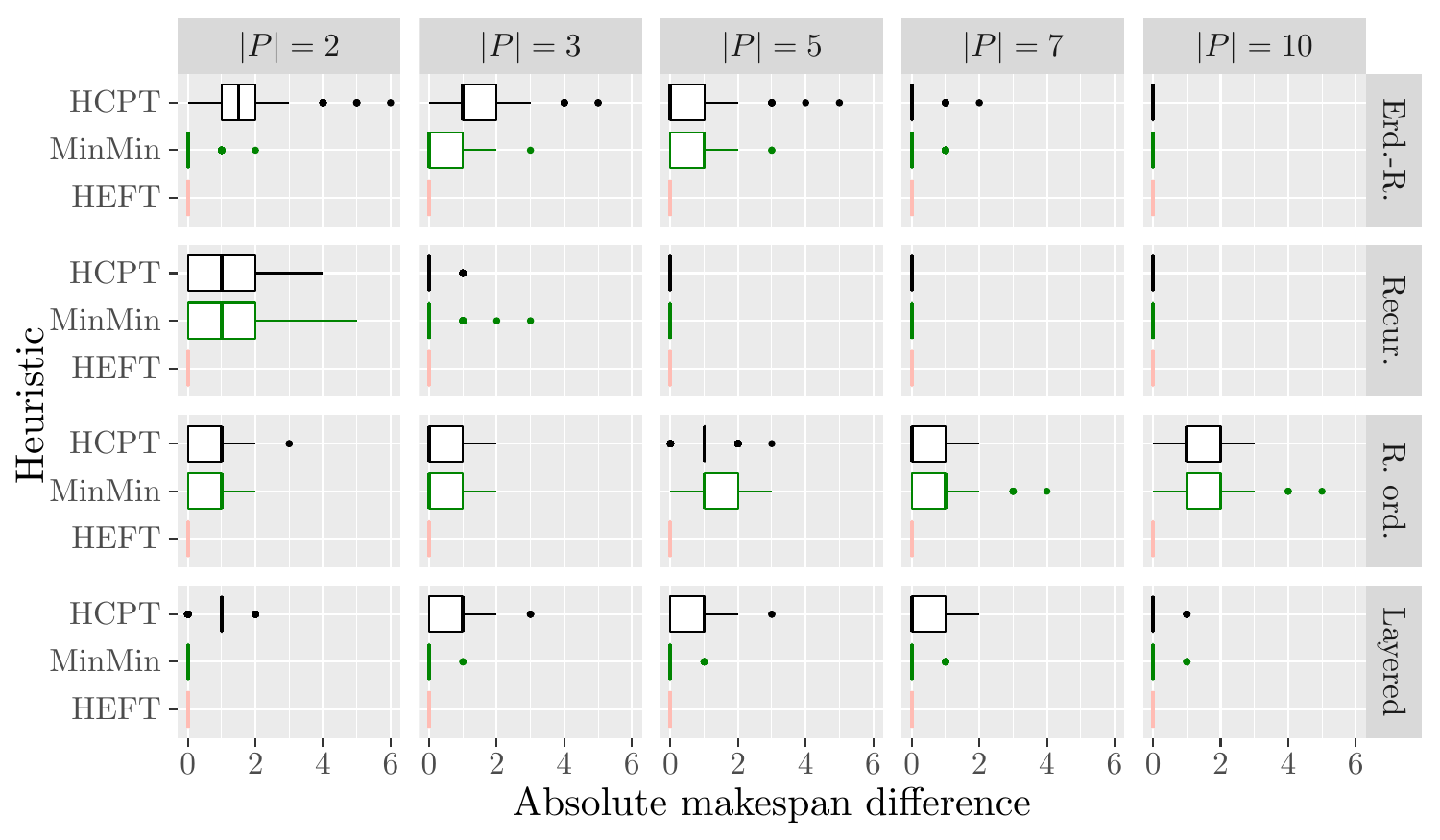}
  \caption{\label{fig:comparison} Difference between the makespan
    obtained with any heuristic and the best value among the three
    heuristics for each instance.
    Each boxplot represents the results for 300 DAGs of size $n=100$
    built with one of the following methods: the \erdos algorithm with
    probability $p=0.15$, the recursive algorithm, the random orders
    algorithm with $K=3$ permutations and the layer-by-layer algorithm
    with probability $p=0.5$ and a number of layers $k=10$.
    Costs are unitary and $|P|$ represents the number of processors.
  }
\end{figure}

Figure~\ref{fig:comparison} shows the absolute difference between
HEFT, HCPT and MinMin for each generation method covered in
Section~\ref{sec:existinggeneration}.
Despite guaranteeing an unbiased generation, instances built with the
recursive algorithm fail to discriminate heuristics except when there
are two processors.
Recall that the mean shape is close to $1.5$ for such DAGs and few
processors are sufficient to obtain a makespan equal to the DAG length
(i.e.\ an optimal schedule).
In contrast, instances built with the random orders algorithm lead to
difference performance for each scheduling heuristics.
However, this generation method has no uniformity guarantee and its
discrete parameter $K$ limits the diversity of generated DAGs.
Finally, the last two algorithms, \erdos and layer-by-layer, fail to
highlight a significant difference between MinMin and HEFT even though
the former scheduling heuristic can be expected to be inferior to the
latter because it discards the DAG structure.

To support these observations, we analyse below the maximum difference
between the makespan obtained with HEFT and the ones obtained with the
other two heuristics.
Because it lacks any backfilling mechanism, HCPT performs worse than
HEFT with an instance composed of the following two elements.
First, a chain of length $k$ with $|P|-1$ additional tasks with
predecessor the $(k-2)$th task of the chain and successor the $k$th
task of the chain.
Alternatively, this first element can be seen as a chain of length
$k-3$ connected to a fork-join with width $|P|$.
The second element is a chain of length $k-1$.
HCPT schedules the first element and then the second one afterward,
leading to a makespan of $2k-1$ whereas the optimal one is $k$.
With $n=100$ tasks and $|P|\le10$, the difference from HEFT with this
instance is greater than or equal to 45.
Moreover, MinMin also performs worse with specific instances.
Consider the ad hoc instances considered in\cite{canon2018a} each
consisting of one chain of length $k$ and a set of $k(|P|-1)$
independent tasks.
Discarding the information about critical tasks prevents MinMin from
prioritizing tasks from the chain.
With $n=100$ tasks and with $|P|\le10$, the worst-case absolute
difference can be greater than or equal to 9 (when MinMin completes
first the independent before starting the chain).
While the difficult instances for HCPT rely on a specific weakness, it
is interesting to analyse the properties of the difficult instances
for MinMin.
Each DAG is characterized by a length equal to $\len=\frac{n}{|P|}$
and a number of edges in the transitive reduction $m(D^T)=\len-1$
(leading to a large width and a large shape standard deviation).
With $n=100$ tasks, with both HCPT and MinMin, the absolute difference
from HEFT can be greater than or equal to 9.

Theses experiments illustrate the need for better generation methods
that control multiple properties while avoiding any generation bias.
An ideal generation method would uniformly select a DAG over all
existing DAGs having a given number of tasks $n$, number of edges $m$
and/or $m(D^T)$, length and/or width, and with a unitary mass.

\section{Conclusion}

This work contributes in three ways to the final objective of
uniformly generating random DAGs belonging to a category of instances
with desirable characteristics.
First, we identify a list of 34 DAG properties and focus on a
selection of 8 such properties.
Among these, the mass quantifies how much an instance can be
decomposed into smaller ones.
Second, existing random generation methods are formally analyzed and
empirically assessed with respect to the selected properties.
Establishing the sub-exponential generic time complexity for
decomposable scheduling problems with uniform instances constitutes
the most noteworthy result of this paper.
Last, we study how the generation methods impact scheduling heuristics
with unitary costs.

The relevance and impact of many other properties need to be
investigated.
For instance, the number of tasks present on a critical path can
exceed the length and even reach $n$.
Also, we could measure the distance of a DAG from a serie-parallel one
by counting with the minimum number of edges to remove in the former
DAG to obtain the latter one.
Both these measures may impact the performance of scheduling
heuristics.

Adapting current results to instances with communication costs
requires some adaptations that need to be explored.
For instance, each edge with a cost could be discarded when there is
another path of higher processing cost (i.e.\ assuming all
communication costs are null on this path).
The definition of the mass could state that a vertex is a bottleneck
vertex when no edge connect a preceding vertex to a following one.

Finally, extending properties to instances with non-unitary costs is
left to future work.
For instance, the shape could be replaced by the continuous occupation
of the DAG when scheduled on an infinite number of processors (i.e.\
the number of occupied processors at each time step).
As a result, the length would be the critical path length and the mean
shape would be the sum of all costs (called the work) divided by the
critical path length (called parallelism in\cite{tobita2002a}).

\bibliographystyle{alpha}
\bibliography{biblio-taskgraph.bib}

\appendix

\section{Exact Properties of Special DAGs}
\label{sec:exact-prop-spec}

Tables~\ref{tab:DAG-prop-prec-edge} and~\ref{tab:DAG-prop-prec-node}
synthesize the exact properties of the special DAGs presented in
Section~\ref{sec:analys-spec-dag}.

\begin{sidewaystable}
  \centering
  \begin{tabular}{m{0.08\columnwidth}ccccccccc}
    \toprule
    DAG & $m$ & $\degree{}{max}$ & $\degree{in}{max}$ & $\degree{out}{max}$
    & $\degree{}{min}$ & $\degree{}{mean}$ & $\degree{}{sd}$ & $\degree{in}{sd}$
    & $\degree{out}{sd}$\\
    \midrule
    $D_{\textnormal{empty}}$ & 0 & 0 & 0 & 0 & 0 & 0 & 0 & 0 & 0\\
    $D_{\textnormal{complete}}$ & $\frac{n(n-1)}{2}$
    & $n-1$ & $n-1$ & $n-1$ & $n-1$ & $n-1$ & 0 & $\sqrt{\frac{n^2-1}{12}}$ & $\sqrt{\frac{n^2-1}{12}}$\\
    $D_{\textnormal{chain}}$ & $n-1$ & $2$ & $1$ & $1$ & $1$ & $2(1-\frac{1}{n})$
    & $\sqrt{\frac{2}{n}(1-\frac{2}{n})}$ & $\sqrt{\frac{1}{n}(1-\frac{1}{n})}$ & $\sqrt{\frac{1}{n}(1-\frac{1}{n})}$\\
    $D_{\textnormal{out-tree}}$ $D_{\textnormal{comb}}$ & $n-1$ & $3$
    & $1$ & $2$ & $1$ & $2(1-\frac{1}{n})$ & $\sqrt{1-\frac1{n}-\frac4{n^2}}$
    & $\sqrt{\frac{1}{n}(1-\frac{1}{n})}$ & $\frac{\sqrt{n-1}\sqrt{n+1}}{n}$\\
    $D_{\textnormal{in-tree}}$ $D_{\textnormal{comb}}^R$ & $n-1$ & $3$ & $2$
    & $1$ & $1$ & $2(1-\frac{1}{n})$ & $\sqrt{1-\frac1{n}-\frac4{n^2}}$
    & $\frac{\sqrt{n-1}\sqrt{n+1}}{n}$ & $\sqrt{\frac{1}{n}(1-\frac{1}{n})}$\\
    $D_{\textnormal{bipartite}}$ & $\frac{n^2}{4}$ & $\frac{n}{2}$ & $\frac{n}{2}$
    & $\frac{n}{2}$ & $\frac{n}{2}$ & $\frac{n}{2}$ & 0 & $\frac{n}{4}$ & $\frac{n}{4}$\\
    $D_{\textnormal{square}}$ & $n(\sqrt{n}-1)$ & $2\sqrt{n}$ & $\sqrt{n}$
    & $\sqrt{n}$ & $\sqrt{n}$ & $2(\sqrt{n}-1)$ & $\sqrt{2\sqrt{n}-4}$
    & $\sqrt{\sqrt{n}-1}$ & $\sqrt{\sqrt{n}-1}$\\
    $D_{\textnormal{triangular}}$ & $\frac{k(k+1)(k-1)}{3}$ & $2(k-1)$ & $k-1$
    & $k$ & 2 & $\frac{4}{3}(k-1)$ & $\frac{\sqrt{2}}{3}(k-1)$
    & $\frac{\sqrt{(k-1)(k+2)}}{3\sqrt{2}}$ & $\frac{\sqrt{(k-1)(k+14)}}{3\sqrt{2}}$\\
    \bottomrule
  \end{tabular}
  \caption{\label{tab:DAG-prop-prec-edge}Edge-related properties of
    special DAGs.
    For $D_{\textnormal{triangular}}$, $k$ is the length of the DAG
    ($k\approx\sqrt{2n}-\frac{1}{2}$ because $n=\frac{k(k+1)}{2}$).}
\end{sidewaystable}

\begin{sidewaystable}
  \centering
  \begin{tabular}{m{0.08\columnwidth}ccccccccc}
    \toprule
    DAG & $\len$ & $\width$ & $\shape{max}$
    & $\shape{min}$ & $\shape{mean}$ & $\shape{sd}$ & $\shape{1}$ & $\shape{k}$ & $\mass$\\
    \midrule
    $D_{\textnormal{empty}}$ & 1 & $n$ & $n$ & $n$ & $n$ & 0 & $n$ & $n$ & 1\\
    $D_{\textnormal{complete}}$ $D_{\textnormal{chain}}$ & $n$ & 1 & 1 & 1 & 1 & 0 & 1 & 1 & 0\\
    $D_{\textnormal{out-tree}}$ & $\log_2(n+1)$ & $\frac{n+1}{2}$ & $\frac{n+1}{2}$ & 1 & $\frac{n}{\log_2(n+1)}$ & $\sqrt{\frac{n}{\log_2(n+1)}\left(\frac{n+2}{3}-\frac{n}{\log_2(n+1)}\right)}$ & 1 & $\frac{n+1}{2}$ & $1-\frac1{n}$\\
    $D_{\textnormal{in-tree}}$ & $\log_2(n+1)$ & $\frac{n+1}{2}$ & $\frac{n+1}{2}$ & 1 & $\frac{n}{\log_2(n+1)}$ & $\sqrt{\frac{n}{\log_2(n+1)}\left(\frac{n+2}{3}-\frac{n}{\log_2(n+1)}\right)}$ & $\frac{n+1}{2}$ & 1 & $1-\frac1{n}$\\
    $D_{\textnormal{comb}}$ & $\frac{n+1}{2}$ & $\frac{n+1}{2}$ & 2 & 1 & $2(1-\frac{1}{n+1})$ & $\sqrt{\frac{2}{n+1}(1-\frac{2}{n+1})}$ & 1 & 2 & $1-\frac1{n}$\\
    $D_{\textnormal{comb}}^R$ & $\frac{n+1}{2}$ & $\frac{n+1}{2}$ & $\frac{n+1}{2}$ & 1 & $2(1-\frac{1}{n+1})$& $\frac{n-1}{n+1}\sqrt{\frac{n-1}{2}}$ & $\frac{n+1}{2}$ & 1 & $\frac1{2}+{1}{2n}$\\
    $D_{\textnormal{bipartite}}$ & 2 & $\frac{n}{2}$ & $\frac{n}{2}$ & $\frac{n}{2}$ & $\frac{n}{2}$ & 0 & $\frac{n}{2}$ & $\frac{n}{2}$ & 1\\
    $D_{\textnormal{square}}$ & $\sqrt{n}$ & $\sqrt{n}$ & $\sqrt{n}$ & $\sqrt{n}$ & $\sqrt{n}$ & 0 & $\sqrt{n}$ & $\sqrt{n}$ & 1\\
    $D_{\textnormal{triangular}}$ & $k$ & $k$ & $k$ & 1 & $\frac{k+1}{2}$ & $\sqrt{\frac{k^2-1}{12}}$ & 1 & $k$ & $1-\frac1{n}$\\
    \bottomrule
  \end{tabular}
  \caption{\label{tab:DAG-prop-prec-node}Vertex-related properties of
    special DAGs.
    For $D_{\textnormal{triangular}}$, $k$ is the length of the DAG
    ($k\approx\sqrt{2n}-\frac{1}{2}$ because $n=\frac{k(k+1)}{2}$).}
\end{sidewaystable}

\section{Probabilistic Properties of Random Triangular Matrices}
\label{sec:annexe:probaER}

We investigate in this section some probabilistic results on DAG
generated by the \erdos approach.

\begin{proposition}\label{prop:ERmean}
  Let $D$ be a DAG with $n$ vertices randomly generated by the \erdos
  algorithm with parameter $p$.
  Denoting by $(X_1,\ldots,X_k)$ its shape decomposition, one has, for
  each $1\leq i\leq k$, $\Expect(|X_i|)\leq\frac{1}{p}$.
\end{proposition}

\begin{proof}
  Let $M_{i,j}$ be the upper triangular matrix corresponding to
  $D$. Let $Y_i=\cup_{j<i} X_j$.  If $j\in X_i$, then, for all $r <j$
  such that $r\notin Y_i$, $M_{r,j}=0$. Therefore, since the
  $M_{r,j}$ are independent Bernoulli random variables of parameter $p$,
  $\Prob(j\in X_i)\leq (1-p)^{j-|Y_i|}$.  Consequently, $\Expect(x_i)\leq
  1+(1-p)+\ldots+(1-p)^{n-|Y_i|}\leq \frac{1}{1-(1-p)}=\frac{1}{p}$.
\end{proof}

\begin{proposition}\label{prop:ERET}
  Let $D$ be a DAG with $n$ vertices randomly generated by the \erdos
  Algorithm with parameter $p$. One has $\Expect(m(D^T))\leq \frac{n-1}{p}-\frac{1-p^2}{p^3}(1-(1-p^2)^{n-1})$.
\end{proposition}

\begin{proof}
Let $M_{i,j}$ be the upper triangular matrix corresponding to $D$.
$$A=\{(i,j)\mid 1\leq i <j \leq n,\ (i,j)\in E\text{ s.t. } \forall
i<r<j,\ (i,r)\notin E \text{ or }(r,j)\notin E\},$$ where $E$ is the set of
edges of $D$. By definition of $D^T$, if $(i,j)\in D^T$, then
$(i,j)\in A$. Consequently,
\begin{equation}\label{eq:prop:ERET1}
  |D^T|\leq |A|.
  \end{equation}

Moreover, for every $i<j$, $\Prob((i,j)\in A)=\Prob(M_{i,j}=1\text{
  and } \forall i<r<j,\ M_{i,r}=0\text{ or } M_{r,j}=0)$.
Since the $M_{i,j}$ are independent Bernoulli random variables,
\begin{align*}
  \Prob((i,j)\in A)&=\Prob(M_{i,j}=1)\Pi_{r=i+1}^{j-1}\Prob( M_{i,r}=0\text{ or } M_{r,j}=0)\\
  &=p \Pi_{r=i+1}^{j-1}(1- \Prob( M_{i,r}=1\text{ and } M_{r,j}=1))\\
  &=p\Pi_{r=i+1}^{j-1}(1-p^2)=p(1-p^2)^{j-i-1}.
\end{align*}

Let $A_{i,j}$ be the Bernoulli random variable encoding that $(i,j)\in A$.
One has $|A|=\sum_{i<j} A_{i,j}$. Consequently,
\begin{align*}
  \Expect(|A|)&=\sum_{j=2}^n\sum_{i=1}^{j-1} \Expect(A_{i,j})=\sum_{j=2}^n\sum_{i=1}^{j-1}p(1-p^2)^{j-i-1}\\
  &=\sum_{j=2}^n\sum_{r=0}^{j-2}p(1-p^2)^{r}\quad \text{with }r=j-i-1\\
  &=p\sum_{j=2}^n\frac{1-(1-p^2)^{j-1}}{1-(1-p^2)}\\
  &=\frac{1}{p}\sum_{j=2}^n(1-(1-p^2)^{j-1})\\
  &=\frac{n-1}{p}-\frac{1}{p}\sum_{j=2}^n(1-p^2)^{j-1}\\
  &=\frac{n-1}{p}-\frac{1-p^2}{p}\sum_{j=0}^{n-2}(1-p^2)^{j}\\
  &=\frac{n-1}{p}-\frac{1-p^2}{p}\frac{1-(1-p^2)^{n-1}}{1-(1-p^2)}\\
  &=\frac{n-1}{p}-\frac{1-p^2}{p^3}(1-(1-p^2)^{n-1}).
\end{align*}

One can conclude using Equation~(\ref{eq:prop:ERET1}).
\end{proof}

\section{Probabilistic Properties of Random Uniform DAGs}
\label{sec:annexe:proba}

We are interested in this section in the probabilistic properties of
the shape of a DAG $D$ randomly generated with the uniform distribution
as exposed in Section~\ref{sec:recursive}. In this context, we
consider a random shape $(x_1,\ldots,x_k)$ generated by
Algorithm~\ref{algo:shaperec} (for a DAG with $n$ vertices).  Note that
the length $k$ of the shape is a random variable and that the $x_i$'s
are dependent random variables. However, the distribution of $x_i$
only depends on the sum of the $x_j$'s, with $j<i$ (formally,
$\Prob(x_i=r\mid x_1,\ldots,x_i-1)=\Prob(x_i=r\mid
s_i)=\frac{a_{n-s_i,k}}{a_{n-s_i}}$). Let $s_1=0$ and for $i\geq 1$,
$s_i=\sum_{j<i}x_j$. The following result is proved
in\cite[Proposition~3]{liskovetsmaixmal}: if $n-s_i\geq 2$ and $r\geq
2$,
\begin{equation}\label{eq:liskovets}
  \Prob(x_i\leq r\mid s_i)\geq
  1-\frac{n-s_i-r}{n-s_i+1}\frac{2}{(r+1)!\ 2^{\frac{r(r-1)}{2}}}.
\end{equation}

It follows from Equation~(\ref{eq:liskovets}), for $r\geq 2$ and if
$n-s_i\geq 2$,
\begin{equation*}
  \Prob(x_i> r\mid s_i)\leq
  \frac{n-s_i-r}{n-s_i+1}\frac{2}{(r+1)!\ 2^{\frac{r(r-1)}{2}}}\leq
  \frac{2}{(r+1)!\ 2^{\frac{r(r-1)}{2}}}.
\end{equation*}

The above equation still holds if $n-s_i< 2$ since
the probability is null.

These upper bounds show that the probability of having large values in
the shape is very small.
For instance, the probability that $x_i\geq 9$ is less than
$10^{-11}$.

Moreover, since for $r >2$, $(r+1)!\geq
2^{r+1}$, one has for every $r>2$,

\begin{equation}\label{eq:preuve1}
  \Prob(x_i> r\mid s_i)\leq
  2^{-\frac{r(r+1)}{2}}.
\end{equation}

The following lemma will be useful.

\begin{lemma}\label{lemma:max}
  One has, for every $r>2$, $n\geq 2$,
  $\Prob(\max(x_i)>r)\leq n2^{-\frac{r(r+1)}{2}}$.
\end{lemma}

\begin{proof}
  Let $A_{i,r}$ denotes the event \emph{$x_i\leq r$}. One has
  \begin{align*}
    &\Prob(\max(x_i)\leq r)=\Prob(\bigcap_{1\leq i \leq k} A_{i,r})\\
    &=\Prob(A_{1,r})\Prob(A_{2,r}\mid A_{1,r})\ldots
      \Prob(A_{i,r}\mid A_{1,r},A_{2,r},\ldots,A_{i-1,r})\ldots
      \Prob(A_{k,r}\mid A_{1,r},A_{2,r},\ldots,A_{k-1,r})
  \end{align*}

  Using Equation~(\ref{eq:preuve1}) and Bernoulli's inequality, it follows that
  $\Prob(\max(x_i)\leq r)\geq (1-2^{-\frac{r(r+1)}{2}})^k\geq
  (1-2^{-\frac{r(r+1)}{2}})^n\geq 1-n2^{-\frac{r(r+1)}{2}}.$
  Now, $\Prob(\max(x_i)>r)=1-\Prob(\max(x_i)\leq r)\leq
  n2^{-\frac{r(r+1)}{2}}$.
\end{proof}

We can now claim an upper bound for the expected value of
$\shape{max}(D)=\max(x_i)$.

\begin{proposition}\label{prop:max}
  One has $\Expect(\max(x_i))=O(\log n)$.
\end{proposition}

\begin{proof}
  Let $h=\left\lfloor \sqrt{6\log_2 n}\right\rfloor+1$.
  \begin{align*}
    \Expect(\max(x_i))&=\sum_{r=1}^n \Prob(\max(x_i)=r).r\\
                         &=\sum _{r=1}^hr\Prob(\max(x_i)=r) \quad+\quad
                           \sum _{r=h+1}^nr\Prob(\max(x_i)=r)\\
                         &\leq \sum _{r=1}^hh\Prob(\max(x_i)=r)+
                           \sum _{r=h+1}^n r \Prob(\max(x_i)=r)\\
    &\leq h^2+\sum _{r=h+1}^n r \Prob(\max(x_i)\geq r)\\
    &\leq h^2+\sum _{r=h+1}^n r \Prob(\max(x_i> r-1)\\
    &\leq h^2+\sum _{r=h+1}^n r n 2^{-\frac{r(r-1)}{2}}\\
    &\leq h^2+n 2^{-\frac{h(h+1)}{2}}\sum _{r=h+1}^n r\\
    &\leq 6\log_2 n +n^3 2^{-\frac{h(h+1)}{2}}
  \end{align*}

  Since $r>3$ for $n\ge2$, we can apply Lemma~\ref{lemma:max} to
  eliminate the probability in the second term.
  Note that $2^{-\frac{h(h+1)}{2}}\leq 2^{-\frac{h^2}{2}}\leq
  2^{-3\log_2 n}$.
  Since $n^3 2^{-3\log_2 n}=1$, we have
  $\Expect(\max(x_i))\leq  6\log_2 n +1$, proving the result.
\end{proof}

It is proved in\cite{liskovetsmaixmal} that $\Expect(x_1)$ converges to a
constant (approximately 1.488) when $n$ grows to infinity. One can
easily obtain a bound for each level and each $n$.

\begin{proposition}\label{prop:expectshape}
  One has, for every $n\geq 2$, every $1\leq i\leq k$,
  $\Expect(x_i)\leq 2+\frac1{4}$.
\end{proposition}

\begin{proof}
  \begin{align*}
    \Expect(x_i)&=\sum_{r=1}^n j \Prob(x_i=r)=
            \Prob(x_i=1)+2\Prob(x_i=2)
            +\sum_{r=3}^n r\Prob(x_i=r)\\
  \end{align*}
  Note that $\Prob(x_i=1)+\Prob(x_i=2)\le1$.
  \begin{align*}
    \Expect(x_i)&\leq 2+\sum_{r=3}^n \frac{r}{2^{r(r+1)/2}} \leq
            2+\sum_{r=3}^n \frac{1}{2^r}\frac{r}{2^{(r+1)/2}}\\
          &\leq 2+\sum_{r=3}^n \frac{1}{2^r}.
  \end{align*}
  Since $r\ge3$ for $n\ge2$, we can apply Lemma~\ref{lemma:max} to
  eliminate the probability in the last term.
  Since $\sum_{r=1}^n \frac{1}{2^r}\leq 1$, $\Expect(x_i)\leq 2+\frac1{4}$.
\end{proof}

Previous results confirm experimental results in
Section~\ref{sec:recursive} and show that the values
of the shape are all quite small. In order to evaluate the mass of a
random DAG, we will now investigate the lengths of the bloc. More
precisely, let
$$\ell_{\max}=\max \{\ell\mid \exists i\text{ s. t. }
(1-x_i)(1-x_{i+1})\ldots(1-x_{i+\ell-1})\neq 0\},$$
the maximum
length of a sequence of consecutive $x_i$ non equal to~$1$.

It is proved in\cite[page 407]{liskovetsmaixmal} that there exists a
constant $0< \alpha_0<2/3$ such that for all $n$, $\Prob(x_i=1\mid
s_i)\geq \alpha_0$ (the constant proposed in\cite{liskovetsmaixmal} is
$\frac{1}{96}$ but practical evaluation leads to claim that the
probability to have only a single vertex in a level is greater than or equal
to $1/3$).

\begin{lemma}\label{lemma:run}
  For every $r>0$,
  $\Prob(\ell_{\max}\geq \ell)\leq n(1-\alpha_0)^\ell$.
\end{lemma}

\begin{proof}
  One has
  \begin{align*}
    \Prob((1-x_i)&(1-x_{i+1})\ldots(1-x_{i+\ell-1})\neq 0\mid s_i)\\
                      &=\Prob((1-x_i)\neq 0\mid s_i)
                        \Prob((1-x_{i+1})\ldots(1-x_{i+\ell-1})\neq
                        0\mid s_i\text{ and } x_1\neq 1)\\
                      &\leq(1-\alpha_0)\Prob((1-x_{i+1})
                        \ldots(1-x_{i+\ell-1})\neq 0\mid s_{i+1})
  \end{align*}
  By a direct induction, one get
  $\Prob((1-x_i)(1-x_{i+1})\ldots(1-x_{i+\ell-1})\neq 0\mid
  s_i)\leq (1-\alpha_0)^\ell$.
  Now let $\overline{x}_i$ be defined, for $1\leq i\leq n$ by
  $\overline{x}_i=x_i$ if $i\leq k$ and $\overline{x}_i=1$ otherwise.
  We also denote by $\overline{s}_i$ the sum $\sum_{j=1}^{-1}
  \overline{x}_i$.
  \begin{align*}
    \Prob(\ell_{\max}\geq \ell)
    &=\Prob(\cup_{i=1}^{k-\ell+1}\{(1-x_i)(1-x_{i+1})
    \ldots(1-x_{i+\ell-1})\neq 0\mid s_i\})\\
    &\leq \Prob(\cup_{i=1}^{n-\ell+1}\{(1-\overline{x}_i)(1-\overline{x}_{i+1})
    \ldots(1-\overline{x}_{i+\ell-1})\neq 0\mid \overline{s}_i\})\\
    &\leq \sum_{i=1}^{k-n+1}\Prob((1-x_i)(1-x_{i+1})
      \ldots(1-x_{i+\ell-1})\neq 0\mid s_i)\\
    &\leq n(1-\alpha_0)^\ell,
  \end{align*}
  proving the result.
\end{proof}

One can now prove Theorem~\ref{theorem:mass}.

\begin{proof}[Proof of Theorem~\ref{theorem:mass}]
  Consider the event $A_n=\shape{max}(D)\geq \log^2(n)$ and
  $B_n=\ell_{\max}(D)\geq\log^2(n)$.  Using Markov Inequality and
  Proposition~\ref{prop:max}, $\Prob(A_n)\leq \frac{1}{\log n}$.
  Therefore, $\Prob(A_n)\to 0$ when $n\to+\infty$.

  Moreover $\Prob(B_n)\leq n(1-\alpha_0)^{\log^2n}$ by
  Lemma~\ref{lemma:run}.
  But $\log(n(1-\alpha_0)^{\log^2n})=\log n + \log(1-\alpha_0) \log^2n$.
  Since $0<1-\alpha_0<1$, $\log n + \log(1-\alpha_0) \log^2n\to -\infty$ when
  $n\to+\infty$. Consequently $n(1-\alpha_0)^{\log^2n}\to 0$ when $n\to+\infty$.

  Therefore $\Prob(A_n\cup B_n)\to 0$ when $n\to +\infty$.
\end{proof}

\end{document}